\documentclass[letterpaper,11pt]{article}

\usepackage{times}

\usepackage{amsthm,amsmath,amsfonts}
\newcommand{\bs}[1]{\boldsymbol{#1}}

\usepackage{graphicx}
\usepackage[small]{caption}
\usepackage[hidelinks]{hyperref}
\newcommand{\orcid}[1]{\,\href{https://orcid.org/#1}{\includegraphics[width=8pt]{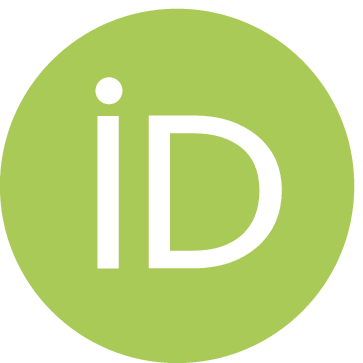}}}

\usepackage{color}
\newcommand{\hl}[1]{\textcolor{red}{#1}}

\newcommand{\I}{\mathcal{I}}
\newcommand{\J}{\mathcal{J}}
\newcommand{\inv}{\mathrm{inv}}

\newtheorem{theorem}{Theorem}
\newtheorem{proposition}{Proposition}
\newtheorem{lemma}{Lemma}
\theoremstyle{remark}\newtheorem*{remark}{Remark}

\title{Moving intervals for packing and covering}

\author{Rain Jiang\orcid{0000-0002-0144-942X}\qquad
Kai Jiang\orcid{0000-0001-8165-0571}\qquad
Minghui Jiang\orcid{0000-0003-1843-9292}\,\thanks{\texttt{ dr.minghui.jiang at gmail.com}}\medskip\\
Home School, USA}

\date{}

\begin{document}

\maketitle

\begin{abstract}
	We study several problems on geometric packing and covering with movement.
	Given a family $\I$ of $n$ intervals of $\kappa$ distinct lengths,
	and another interval $B$,
	can we pack the intervals in $\I$ inside $B$
	(respectively, cover $B$ by the intervals in $\I$)
	by moving $\tau$ intervals
	and keeping the other $\sigma = n - \tau$ intervals unmoved?
	We show that both packing and covering are W[1]-hard with any one of
	$\kappa$, $\tau$, and $\sigma$ as single parameter,
	but are FPT with combined parameters $\kappa$ and $\tau$.
	We also obtain improved polynomial-time algorithms for packing and covering,
	including an $O(n\log^2 n)$ time algorithm for covering,
	when all intervals in $\I$ have the same length.
\end{abstract}

\section{Introduction}

Let $\I$ be a family of $n$ intervals of $\kappa$ distinct lengths
and of total length $\ell_\I$.
Let $B$ be another interval of length $\ell_B$,
in the same line as the intervals in $\I$.
We study the following problems of deciding,
for two parameters $\sigma$ and $\tau$ where $\sigma + \tau = n$,
whether we can move $\tau$ intervals in $\I$,
and keep the other $\sigma$ intervals in $\I$ unmoved,
to achieve certain geometric configurations:
\begin{align*}
	\textsc{Pack}&:
	\textrm{the $n$ intervals in $\I$ are pairwise-disjoint and are contained in $B$,}\\
	\textsc{Cover}&:
	\textrm{the union of the $n$ intervals in $\I$ contains $B$,}\\
	\textsc{Join}&:
	\textrm{the $n$ intervals in $\I$ are joined into a contiguous interval of length $\ell_\I$,}\\
	\textsc{J-Pack}&:
	\textrm{the $n$ intervals in $\I$ are joined into a contiguous interval of length $\ell_\I$ that is contained in $B$,}\\
	\textsc{J-Cover}&:
	\textrm{the $n$ intervals in $\I$ are joined into a contiguous interval of length $\ell_\I$ that contains $B$,}\\
	\textsc{Tile}&:
	\textrm{the $n$ intervals in $\I$ are joined into a contiguous interval of length $\ell_\I$ that coincides with $B$.}
\end{align*}

Without loss of generality, we assume that
$\ell_\I \le \ell_B$ for
\textsc{Pack} and
\textsc{J-Pack},
$\ell_\I \ge \ell_B$ for
\textsc{Cover} and \textsc{J-Cover},
and
$\ell_\I = \ell_B$ for \textsc{Tile}.
The problem \textsc{Tile} is a special case of the four problems
\textsc{Pack}, \textsc{Cover}, \textsc{J-Pack}, and \textsc{J-Cover},
which are equivalent when $\ell_\I = \ell_B$.
The problem \textsc{Join} is similar to the problem \textsc{Tile},
but does not have the constraint of the interval $B$.
Indeed \textsc{Join} can be viewed as a special case of \textsc{J-Pack} where
$B$ contains all intervals in $\I$ and has extra spaces of $\ell_\I$ at both ends.

We represent all intervals in the form $[x, y)$,
which is a set closed at left endpoint $x$ and open at right endpoint $y$.
To avoid unnecessary technicality,
we allow zero-distance moves. That is, an interval may be ``moved''
without changing its position.
This ensures natural monotonicities of $\tau$ and $\sigma$ for these problems.
For example,
denote by \textsc{Pack$(\tau,\sigma)$}
the predicate whether the $n$ intervals in $\I$
can be packed inside the interval $B$
by moving $\tau$ intervals and keeping $\sigma = n - \tau$ intervals unmoved.
If \textsc{Pack$(\tau,\sigma)$} is true and $\sigma > 0$,
then \textsc{Pack$(\tau+1,\sigma-1)$} is also true.
Let
\textsc{Pack$^*$}, \textsc{Cover$^*$},
\textsc{Join$^*$},
\textsc{J-Pack$^*$}, \textsc{J-Cover$^*$},
and \textsc{Tile$^*$}
be the optimization versions of
these decision problems,
for minimizing $\tau$ and maximizing $\sigma$,
and let $\tau^*$ and $\sigma^*$ denote the optimal values,
where $\tau^* + \sigma^* = n$.

In slightly different formulations,
our problems have been studied previously by Mehrandish, Narayanan, and Opatrny~\cite{MNO11},
for the application of intrusion detection and barrier coverage by sensors.
When the number $\kappa$ of distinct lengths of intervals in $\I$ is unrestricted,
Mehrandish et~al.\ proved
that these problems are all weakly NP-hard
by reductions from \textsc{Partition}~\cite[Problem~SP12]{GJ79}.
In particular,
\textsc{Join} is hard
already for $\sigma = 2$~\cite[Theorem~2]{MNO11},
and \textsc{Tile}, hence
\textsc{Pack}, \textsc{Cover}, \textsc{J-Pack}, and \textsc{J-Cover},
are hard already for $\sigma = 1$~\cite[Theorem~3]{MNO11}.
The complexities of these problems for fixed $\kappa$ or $\tau$, however,
have not been examined.

Using standard techniques of dynamic programming, we show that these problems
admit polynomial-time algorithms when either $\kappa$ or $\tau$ is constant:

\begin{proposition}\label{prp:fpt}
	\textsc{Pack},
	\textsc{Join},
	\textsc{J-Pack},
	\textsc{J-Cover},
	and \textsc{Tile}
	admit algorithms running in
	$O(n\log n + n\cdot f(\kappa,\tau))$
	time,
	and
	\textsc{Cover}
	admits an algorithm running in
	$O(n^2\cdot f(\kappa,\tau))$
	time,
	for some function $f$
	that is bounded by a polynomial in $n$ when either $\kappa$ or $\tau$ is constant.
\end{proposition}

In terms of parameterized complexity,
Proposition~\ref{prp:fpt} implies that these problems are fixed-parameter tractable (FPT)
with combined parameters $\kappa$ and $\tau$.
Extending the previous hardness results,
we prove the following theorem:

\begin{theorem}\label{thm:w1hard}
	\textsc{Pack}, \textsc{Cover}, \textsc{Join}, \textsc{J-Pack}, \textsc{J-Cover},
	and \textsc{Tile}
	are NP-hard,
	and are W[1]-hard
	with any one of $\kappa$, $\tau$, and $\sigma$ as single parameter,
	even when all interval coordinates are integers encoded in unary.
\end{theorem}

When all intervals in $\I$ have the same length,
that is, when $\kappa=1$,
Mehrandish et~al.~\cite[Theorems~4, 7, 10, 11]{MNO11}
presented an $O(n^2)$ time algorithm for
\textsc{Join$^*$}
and \textsc{J-Pack$^*$},
and an $O(n^3)$ time algorithm
for \textsc{Pack$^*$}
and \textsc{Cover$^*$}.
We obtain improved algorithms for these problems:

\begin{proposition}\label{prp:j1}
	When $\kappa = 1$,
	\textsc{Join$^*$}, \textsc{J-Pack$^*$}, \textsc{J-Cover$^*$}, and \textsc{Tile$^*$}
	admit algorithms running in $O(n\log n)$ time.
\end{proposition}

\begin{theorem}\label{thm:pc1}
	When $\kappa = 1$,
	\textsc{Pack$^*$} and \textsc{Cover$^*$}
	admit algorithms running in $O(n\log n + n\min\{\sigma^*,\tau^*\})$ time,
	and \textsc{Cover$^*$} admits
	an algorithm running in $O(n\log^2 n)$ time.
\end{theorem}

We leave two open questions:
\begin{itemize}\setlength\itemsep{0pt}

		\item
			Is there an exact algorithm running in $O(n\log^2 n)$ time
			for \textsc{Pack$^*$} on intervals of the same length?

		\item
			Are there FPT algorithms with combined parameters $\kappa$ and $\sigma$
			for \textsc{Pack} and \textsc{Cover}?

\end{itemize}

\section{Algorithms on intervals of the same length for
\textsc{J-Pack$^*$} and \textsc{J-Cover$^*$}}

In this section we prove Proposition~\ref{prp:j1} on the four problems
\textsc{Join$^*$},
\textsc{J-Pack$^*$}, \textsc{J-Cover$^*$},
and \textsc{Tile$^*$}.
Recall that
\textsc{Join$^*$} is a special case of \textsc{J-Pack$^*$},
and
\textsc{Tile$^*$} is a special case of
\textsc{J-Pack$^*$} and \textsc{J-Cover$^*$}.
Thus it suffices to present algorithms for
\textsc{J-Pack$^*$} and \textsc{J-Cover$^*$}.

For any real number $x$,
denote by $\{x\} = x - \lfloor x \rfloor$
the \emph{fractional part} of $x$.
For any interval $[x,x+\ell)$ of integer length $\ell$,
define its \emph{fractional coordinate} as $\{x\}$.

Let $B = [0,\ell_B)$,
and $\I = \{ I_1,\ldots,I_n \}$,
where $I_i = [x_i, x_i + 1)$.
Then $\ell_\I = n$.
For $\ddot x \in [0,1)$,
let $\I(\ddot x)$ be the subfamily of intervals in $\I$
with fractional coordinate $\ddot x$.
As in~\cite[Theorem~4 and Theorem~7]{MNO11},
the problem
\textsc{J-Pack$^*$} (respectively, \textsc{J-Cover$^*$})
reduces to computing,
for every distinct fractional coordinate $\ddot x \in [0,1)$ of the intervals in $\I$,
the maximum number of pairwise-disjoint intervals in $\I(\ddot x)$
that are inside some interval $B_{\ddot x}$
of fractional coordinate $\ddot x$ and of length $\ell_\I = n$,
such that $B_{\ddot x}\subseteq B$ (respectively, $B_{\ddot x}\supseteq B$).
Then the intervals in $\I$ can be joined into a contiguous interval
that coincides with $B_{\ddot x}$.

The algorithm for \textsc{J-Pack$^*$} (respectively, \textsc{J-Cover$^*$})
works as follows.
First sort the intervals in $\I$ lexicographically,
by considering each interval $I_i$ as a pair of numbers $(\{x_i\}, x_i)$.
Then the intervals in each subfamily $\I(\ddot x)$ appear consecutively
in the sorted list.
Next scan each subfamily $\I(\ddot x)$ independently.

For each $\ddot x$,
regardless of the intervals in $\I(\ddot x)$,
there exists an interval $B_{\ddot x}\subseteq B$ (respectively, $B_{\ddot x}\supseteq B$)
of fractional coordinate $\ddot x$ and length $n$
if and only if
either $\ddot x = 0$ or
$[\ddot x, \ddot x + n)\subseteq[0,\ell_B) 
$
(respectively, $[\ddot x - 1, \ddot x - 1 + n)\supseteq[0,\ell_B)$),
which can be checked in constant time.
If this condition is satisfied,
then proceed to find the maximum number of pairwise-disjoint intervals in $\I(\ddot x)$
that are contained in $[0, \ell_B)$ (respectively, in $[\ell_B - n, n)$)
such that the span $x_j - x_i + 1$
of the first interval $I_i$ and the last interval $I_j$
is at most $n$.
This can be done in linear time using the standard ``two pointers'' technique.
These intervals in $\I(\ddot x)$, with span at most $n$,
are inside some interval $B_{\ddot x}$ of length exactly $n$
such that $B_{\ddot x}\subseteq[0, \ell_B)$
(respectively, $B_{\ddot x}\subseteq[\ell_B - n, n)$ and hence $B_{\ddot x}\supseteq[0,\ell_B)$).

The overall running time of the algorithm is $O(n\log n)$, which is dominated by the sorting.
This completes the proof of Proposition~\ref{prp:j1}.

\bigskip
\begin{remark}
	We have used fractional numbers for the interval coordinates so that
	the intervals in $\I$ have a convenient unit length of $1$.
	Alternatively, we could let $B$ and the intervals in $\I$ have integer coordinates,
	and let the uniform length $\ell$ of the intervals in $\I$ be an integer greater than $1$.
	In this formulation, we would define $\{x\} = x \bmod \ell$.
	Then the same idea yields an algorithm of the same $O(n\log n)$ running time.
\end{remark}

\section{Algorithms on intervals of the same length for
\textsc{Pack$^*$} and \textsc{Cover$^*$}}

In this section we prove
Theorem~\ref{thm:pc1} on the two problems
\textsc{Pack$^*$} and \textsc{Cover$^*$}.

Let $B = [0,\ell_B)$,
and $\I = \{ I_1,\ldots,I_n \}$,
where $I_i = [x_i, x_i + 1)$.
Then $\ell_\I = n$.
Add two dummy intervals
$I_0 = [x_0, x_0 + 1) = [-1, 0)$
and
$I_{n+1} = [x_{n+1}, x_{n+1} + 1) = [\ell_B, \ell_B + 1)$.
Any interval in $\I$ that does not intersect $B$ can be relocated to either $[-1, 0)$
or $[\ell_B, \ell_B + 1)$ without affecting the answer to either problem.
After relocating outside intervals and sorting,
we can assume that
$$
-1 = x_0 \le x_1 \le x_2 \le \ldots \le x_n \le x_{n+1} = \ell_B.
$$

When comparing pairs,
we use notations such as $<,\le,>,\ge,\min,\max$,
and use terms such as \emph{largest} and \emph{smallest},
in terms of lexicographic order.
For example, $(1, 2) < (2, 1)$, and $\min\{ (1, 2), (2, 1) \} = (1, 2)$.

\subsection{Algorithms for \textsc{Pack$^*$}}

For $0 \le i < j \le n + 1$,
write $i \prec j$ if $x_j - x_i \ge 1$.
For $0 \le j \le n + 1$,
let $l_j$ be the largest index $i$ such that $0 \le i \prec j$,
or $-1$ if no such $i$ exists.
With increasing $x_i$ for $0 \le i \le n + 1$,
$l_i$ for all $i$
can be computed in $O(n)$ time using the standard ``two pointers'' technique.

Recall that
$\{x\} = x - \lfloor x \rfloor$.
For $0 \le i < j \le n + 1$,
define
$$
w(i, j)=
\begin{cases}
	1&\textrm{if $\{x_i\} > \{x_j\}$},\\
	0&\textrm{otherwise}.
\end{cases}
$$
It is easy to check that
the $n$ intervals in $\I$ can be packed inside $B$
by moving $n - \sigma$ intervals and keeping $\sigma$ intervals unmoved
if and only if
there is a sequence $\bs s$ of $\sigma + 2$ indices $s_h$, $0 \le h \le \sigma+1$,
where
$$
0 = s_0 \prec s_1 \prec \ldots \prec s_{\sigma+1} = n + 1,
$$
such that the following condition holds:
$$
\sum_{h=0}^\sigma
\left(
\lfloor x_{s_{h+1}} \rfloor - \lfloor x_{s_h} \rfloor
- w(s_h, s_{h+1})
\right)
\ge n + 1,
$$
which simplifies to
\begin{align*}
	\lfloor x_{n + 1} \rfloor - \lfloor x_0 \rfloor - \sum_{h=0}^\sigma w(s_h, s_{h+1})
	&\ge n + 1\\
	\lfloor \ell_B \rfloor + 1 - \sum_{h=0}^\sigma w(s_h, s_{h+1})
	&\ge n + 1\\
	\sum_{h=0}^\sigma w(s_h, s_{h+1})
	&\le \lfloor \ell_B \rfloor - n.
\end{align*}

Let $\bs\pi$ be a permutation of the $n + 2$ indices $0, 1, \ldots, n, n + 1$
such that for $0 \le i < j \le n + 1$,
$\pi_i > \pi_j$ if and only if
$\{x_i\} > \{x_j\}$,
which can be found by sorting in $O(n\log n)$ time.
For $0 \le i < j \le n + 1$,
denote by $\inv(i, j)$ the indicator variable which is $1$ when $\pi_i > \pi_j$
and is $0$ otherwise.
Then $w(i, j) = \inv(i, j)$.

From this new perspective,
the problem of determining the maximum value $\sigma^*$,
such that
the $n$ intervals in $\I$ can be packed inside $B$
by moving $n - \sigma^*$ intervals and keeping $\sigma^*$ intervals unmoved,
reduces to
the problem of
finding the maximum length $\sigma^* + 2$ of a sequence $\bs s$ of indices
$0 = s_0 \prec s_1 \prec \ldots \prec s_{\sigma^*+1} = n + 1$,
such that
the number of inversions $\pi_i > \pi_j$ between consecutive indices $i < j$ in $\bs s$
(which we will call \emph{drops})
is at most $d^* = \lfloor \ell_B \rfloor - n$.

\paragraph{Dynamic programming}

For $0 \le h \le i \le n$,
define $dp(h,i)$ as the pair $(d(h,i), p(h,i))$,
where $d(h,i)$ is the minimum number of drops
in any sequence $\bs s$ of $h + 1$ indices
$0 = s_0 \prec \ldots \prec s_h \le i$,
and $p(h,i)$ is the minimum value of $p = \pi_{s_h}$
over all such sequences $\bs s$ with exactly $d = d(h,i)$ drops.
If there are no such sequences $\bs s$,
define $dp(h,i) = (\infty,0)$.
For convenience,
also define $dp(h,i) = (\infty,0)$ for $h > i$.

The table $dp$ can be computed by dynamic programming.
For the base case when $h = 0$ or $h = i$,
let $dp(0,i) = (0, \pi_0)$
for $0 \le i \le n$,
and, for $1 \le i \le n$, let
$$
dp(i,i) =
\begin{cases}
	dp(i - 1, i - 1)\to i & \textrm{if $i - 1 \prec i$},\\
	(\infty, 0) & \textrm{otherwise}.
\end{cases}
$$
where
$$
(d, p) \to i =
\begin{cases}
	(\infty, 0) & \textrm{if $(d, p) = (\infty, 0)$},\\
	(d, \pi_i) & \textrm{else if $p < \pi_i$},\\
	(d + 1, \pi_i) & \textrm{otherwise}.
\end{cases}
$$

Then for $0 < h < i \le n$,
in particular, for $h = 1,\ldots,n - 1$ and $i = h+1,\ldots,n$,
we can compute $dp(h,i)$ using the recurrence
$$
dp(h,i) = \min\{\, dp(h, i-1),\, dp(h-1, l_i) \to i \,\}.
$$

Note that $\sigma^*$ is the maximum $h$
such that the $d$-part of
$dp(h, l_{n + 1})\to n+1$
is at most $d^*$,
which can be found in $O(n \sigma^*)$ time by a sequential search
for increasing values of $h$.

We can also find $\tau^* = n - \sigma^*$ in $O(n \tau^*)$ time by another sequential search,
for increasing values of $t = i - h$ instead of $h$.
First compute $dp(h,i)$ for the base case when $h = 0$ or $h = i$ as before,
then compute $dp(h,h+t)$
for $t = 1, \ldots, n - 1$ and $h = 1,\ldots,n - t$.
Note that $\tau^*$ is the minimum $t$
such that the $d$-part of
$dp(n-t, l_{n + 1})\to n+1$
is at most $d^*$.

Thus we have an algorithm for \textsc{Pack$^*$}
running in $O(n\log n + n\min\{\sigma^*,\tau^*\})$ time.

\subsection{Algorithms for \textsc{Cover$^*$}}

Recall that
$\{x\} = x - \lfloor x \rfloor$.
For $0 \le i < j \le n + 1$,
define
$$
w(i, j)=
\begin{cases}
	1&\textrm{if $\{x_i\} < \{x_j\}$ or $x_i = x_j$},\\
	0&\textrm{otherwise}.
\end{cases}
$$
It is easy to check that
$B$ can be covered by the $n$ intervals in $\I$
by moving $n - \sigma$ intervals and keeping $\sigma$ intervals unmoved
if and only if
there is a sequence $\bs s$ of $\sigma + 2$ indices $s_h$, $0 \le h \le \sigma+1$,
where
$$
0 = s_0 < s_1 < \ldots < s_{\sigma+1} = n + 1,
$$
such that the following condition holds:
$$
\sum_{h=0}^\sigma
\left(
\lfloor x_{s_{h+1}} \rfloor - \lfloor x_{s_h} \rfloor
+ w(s_h, s_{h+1})
\right)
\le n + 1,
$$
which simplifies to
\begin{align*}
	\lfloor x_{n + 1} \rfloor - \lfloor x_0 \rfloor + \sum_{h=0}^\sigma w(s_h, s_{h+1})
	&\le n + 1\\
	\lfloor \ell_B \rfloor + 1 + \sum_{h=0}^\sigma w(s_h, s_{h+1})
	&\le n + 1\\
	\sum_{h=0}^\sigma w(s_h, s_{h+1})
	&\le n - \lfloor \ell_B \rfloor.
\end{align*}

We prove a technical lemma in the following:
\begin{lemma}\label{lem:sorting}
	For $0 \le i < j \le n + 1$,
	$(\{x_i\}, -x_i, i) < (\{x_j\}, -x_j, j)$
	if and only if
	$\{x_i\} < \{x_j\}$ or $x_i = x_j$.
\end{lemma}

\begin{proof}
	We first prove the \emph{only if} implication.
	Suppose that
	$(\{x_i\}, -x_i, i) < (\{x_j\}, -x_j, j)$.
	Then by the lexicographic order,
	there are three cases:
	either $\{x_i\} < \{x_j\}$,
	or $\{x_i\} = \{x_j\}$ and $x_i > x_j$,
	or $\{x_i\} = \{x_j\}$ and $x_i = x_j$ and $i < j$.
	Recall that $x_i \le x_j$ for $i < j$.
	Thus the second case does not hold.
	The third case simplifies to $x_i = x_j$.
	In summary, we have
	either $\{x_i\} < \{x_j\}$
	or $x_i = x_j$.

	We next prove the \emph{if} implication.
	Suppose that $\{x_i\} < \{x_j\}$ or $x_i = x_j$.
	If $\{x_i\} < \{x_j\}$,
	then clearly $(\{x_i\}, -x_i, i) < (\{x_j\}, -x_j, j)$.
	If $x_i = x_j$,
	then $\{x_i\} = \{x_j\}$ and $-x_i = -x_j$,
	but since $i < j$, we again have
	$(\{x_i\}, -x_i, i) < (\{x_j\}, -x_j, j)$.
\end{proof}

Let $\bs\pi$ be a permutation of the $n + 2$ indices $0, 1, \ldots, n, n + 1$
such that for $0 \le i < j \le n + 1$,
$\pi_i > \pi_j$ if and only if
$(\{x_i\}, -x_i, i) < (\{x_j\}, -x_j, j)$,
which can be found by sorting in $O(n\log n)$ time.
For $0 \le i < j \le n + 1$,
denote by $\inv(i, j)$ the indicator variable which is $1$ when $\pi_i > \pi_j$
and is $0$ otherwise.
Then, by Lemma~\ref{lem:sorting},
$w(i, j) = \inv(i, j)$.

From this new perspective,
the problem of determining the maximum value $\sigma^*$,
such that
$B$ can be covered by the $n$ intervals in $\I$
by moving $n - \sigma^*$ intervals and keeping $\sigma^*$ intervals unmoved,
reduces to
the problem of
finding the maximum length $\sigma^* + 2$ of a sequence $\bs s$ of indices
increasing from $0$ to $n + 1$,
such that
the number of inversions $\pi_i > \pi_j$ between consecutive indices $i < j$ in $\bs s$
(which we will call \emph{drops})
is at most $d^* = n - \lfloor \ell_B \rfloor$.

\paragraph{Dynamic programming}

For $0 \le h \le i \le n$,
define $dp(h,i)$ as the pair $(d(h,i), p(h,i))$,
where $d(h,i)$ is the minimum number of drops
in any sequence $\bs s$ of $h + 1$ indices
$0 = s_0 < \ldots < s_h \le i$,
and $p(h,i)$ is the minimum value of $p = \pi_{s_h}$
over all such sequences $\bs s$ with exactly $d = d(h,i)$ drops.

The table $dp$ can be computed by dynamic programming.
For the base case when $h = 0$ or $h = i$,
let $dp(0,i) = (0, \pi_0)$
for $0 \le i \le n$,
and let $dp(i,i) = dp(i-1,i-1)\to i$
for $1 \le i \le n$,
where
$$
(d, p) \to i =
\begin{cases}
	(d, \pi_i) & \textrm{if $p < \pi_i$},\\
	(d + 1, \pi_i) & \textrm{otherwise}.
\end{cases}
$$
Then for $0 < h < i \le n$,
in particular, for $h = 1,\ldots,n - 1$ and $i = h+1,\ldots,n$,
we can compute $dp(h,i)$ using the recurrence
$$
dp(h,i) = \min\{\, dp(h, i-1),\, dp(h-1, i-1) \to i \,\}.
$$

Note that $\sigma^*$ is the maximum $h$
such that the $d$-part of
$dp(h, n)\to n+1$
is at most $d^*$,
which can be found in $O(n \sigma^*)$ time by a sequential search
for increasing values of $h$.

We can also find $\tau^* = n - \sigma^*$ in $O(n \tau^*)$ time by another sequential search,
for increasing values of $t = i - h$ instead of $h$.
First compute $dp(h,i)$ for the base case when $h = 0$ or $h = i$ as before,
then compute $dp(h,h+t)$
for $t = 1, \ldots, n - 1$ and $h = 1,\ldots,n - t$.
Note that $\tau^*$ is the minimum $t$
such that the $d$-part of
$dp(n-t,n)\to n+1$
is at most $d^*$.

Thus we have an algorithm for \textsc{Cover$^*$}
running in $O(n\log n + n\min\{\sigma^*,\tau^*\})$ time.

\paragraph{Lagrangian relaxation}

We next present an algorithm for \textsc{Cover$^*$} running in $O(n\log^2 n)$ time.
To do this, we first redesign the dynamic programming algorithm,
then speed it up using Lagrangian relaxation.
The technique of Lagrangian relaxation is also known as \textbf{Aliens trick}
in the competitive programming community.
Our use of this trick here is inspired
by a related problem\footnote{The problem can be translated to the following:
Given a permutation $\bs\pi$ of $[1,n]$,
compute for each $0 \le d < n$
the maximum length of a sequence $\bs s$ of indices
such that the number of inversions $\pi_i > \pi_j$ between consecutive indices $i < j$ in $\bs s$
is at most $d$.
One of the solutions by the problem setters
uses Aliens trick and runs in $O(n\sqrt n\log n)$ time.}
created by Compton and Qi~\cite{CQ21}.

For $0 \le d \le i \le n + 1$,
let $\sigma(d,i)$ be the maximum $h$
such that there exists a sequence $\bs s$ of $h + 2$ indices
increasing from $0$ to $i$, $0 = s_0 < \ldots < s_{h + 1} = i$, with at most $d$ drops.
Let $\sigma(d,i) = -\infty$ if such a sequence does not exist.
For convenience, also define $\sigma(d, i) = -\infty$ for $d > i$.
Our goal is to get $\sigma^* = \sigma(d^*, n + 1)$,
where $d^* = n - \lfloor \ell_B \rfloor \le n$.

We can compute $\sigma^*$ by dynamic programming.
For $d = 0$,
$\sigma(0, i)$
is $2$ less than the length of a longest increasing subsequence of $\pi_0\ldots \pi_i$
starting at $\pi_0$ and ending at $\pi_i$, or $-\infty$ if it does not exist.
In particular, $\sigma(0, 0) = -1$.
With the help of some data structures such as balanced search trees,
$\{ \sigma(0, i) \mid 0 \le i \le n + 1 \}$
can be computed in $O(n\log n)$ time.
If $d^* = 0$,
then we already have $\sigma^* = \sigma(0, n+1)$.
In the following, we assume that $d^* \ge 1$.

For $d = 1,\ldots,d^*$,
we can compute $\{ \sigma(d, i) \mid d \le i \le n + 1 \}$
from $\{ \sigma(d - 1, i) \mid d - 1 \le i \le n + 1 \}$.
For the base case,
$\sigma(d, 0) = \sigma(d - 1, 0)$.
For $1 \le j \le n + 1$,
$$
\sigma(d, j) = \max\big\{\;
\sigma(d - 1, j),\;
\max\{\, \sigma(d - \inv(i, j), i) + 1 \mid 0 \le i < j \,\}
\;\big\}.
$$
With balanced search trees, this takes $O(n \log n)$ time for each $d$.
Thus we can get $\sigma^* = \sigma(d^*, n + 1)$ in $O(d^* n \log n)$ time.

We next speed up the algorithm using Aliens trick,
which depends on a property of $\sigma$ stated in the following lemma.
We say that a function $f:\mathbb{Z}\to\mathbb{Z}$
is \emph{concave} for $a \le x \le b$
if $f(x + 1) - f(x) \le f(x) - f(x - 1)$ for $a < x < b$.

\begin{lemma}\label{lem:concave}
	$\sigma(d, n + 1)$ is concave for $1 \le d \le n + 1$.
\end{lemma}

The trick is to assign a penalty $\lambda$ for each drop.
Instead of computing $\sigma(d, i)$ sequentially for increasing $d$
from $0$ to $d^*$ as in the dynamic programming algorithm,
we will compute $\sigma(d, i) - d\lambda$,
and get $\sigma^* = \sigma(d^*, n + 1)$ indirectly
as $\sigma(d^*, n + 1) - d^*\lambda + d^*\lambda$,
by a binary search for a suitable $\lambda$.

For any $\lambda \ge 0$, and for $1 \le i \le n + 1$,
let
$$
\sigma_\lambda(i) = \max\{\, \sigma(d, i) - d\lambda \mid 1 \le d \le i \,\},
$$
and correspondingly,
let $d_\lambda(i)$ be the minimum $d$, $1 \le d \le i$,
such that
$\sigma(d, i) - d\lambda = \sigma_\lambda(i)$.

Consider $\sigma(d, n + 1) - d\lambda$ as a function of $d$ and $\lambda$.
Geometrically,
it can be viewed as the dot product of two vectors
$(d, \sigma(d, n + 1))$ and $(-\lambda, 1)$.
As $\lambda$ decreases from $\infty$ to $0$,
the vector $(-\lambda, 1)$ rotates from $-\vec x$ direction to $+\vec y$ direction.
For any fixed $\lambda$, the dot products for different $d$
are proportional to the projection lengths of 
different vectors $(d, \sigma(d, n + 1))$ onto the same vector $(-\lambda, 1)$.

For each $d$, $1 \le d \le n$,
equating the two dot products for $d$ and $d + 1$
yields an equation of $\lambda$,
$$
\sigma(d, n + 1) - d\lambda = \sigma(d + 1, n + 1) - (d + 1)\lambda,
$$
which has an integer solution,
$$
\lambda_d = \sigma(d + 1, n + 1) - \sigma(d, n + 1).
$$

By Lemma~\ref{lem:concave},
$\lambda_d$ is decreasing for $1 \le d \le n$.
Since any sequence with at most $d$ drops is also a sequence with at most $d + 1$ drops,
we have
$\sigma(d, n + 1) \le \sigma(d + 1, n + 1)$
and hence
$\lambda_d \ge 0$.
Also, since $\sigma(d, n + 1) + 2 \ge 2$ for $d \ge 1$, and $\sigma(d + 1, n + 1) + 2 \le n + 2$,
we have
$\lambda_d \le n$.
Thus
$$
0 \le \lambda_n \le \ldots \le \lambda_1 \le n.
$$

Thus for each $d$, $1 \le d \le n + 1$,
there is a non-empty range of integer values for $\lambda$,
between $0$ and $n$,
such that
$$
\sigma_\lambda(n + 1)
=
\sigma(d, n + 1) - d\lambda.
$$
Specifically, these ranges are $\lambda_1 \le \lambda \le n$ for $d = 1$,
$\lambda_d \le \lambda \le \lambda_{d - 1}$ for $2 \le d \le n$,
and $0 \le \lambda \le \lambda_n$ for $d = n + 1$.

Correspondingly,
$d_\lambda(n + 1)$ is decreasing for $0 \le \lambda \le n$:
it is equal to
$n + 1$ for $0 \le \lambda < \lambda_n$,
to $d = n,\ldots,2$ for $\lambda_d \le \lambda < \lambda_{d - 1}$,
and
to $1$ for $\lambda_1 \le \lambda \le n$.
Some of these ranges may be empty,
so $d_\lambda(n + 1)$ may not assume every integer value from $n + 1$ down to $1$,
as $\lambda$ increases from $0$ to $n$.
But we always have $d_\lambda(n + 1) = 1$ for $\lambda = n$.
Recall our assumption that $d^* \ge 1$.
Thus $d_\lambda(n + 1) \le d^*$ for $\lambda = n$.

Let $\lambda^*$ be
the smallest integer $\lambda$, $0 \le \lambda \le n$,
such that
$d_\lambda(n + 1) \le d^*$.
Then $\lambda^* = \lambda_{d^*}$,
and hence $\sigma_{\lambda^*}(n + 1) = \sigma(d^*, n + 1) - d^* \lambda^*$.
Then $\sigma^* = \sigma(d^*, n + 1) = \sigma_{\lambda^*}(n + 1) + d^* \lambda^*$.

It remains to find $\lambda^*$.
Recall that \{$\sigma(0, i) \mid 0 \le i \le n + 1$\}
can be computed in $O(n \log n)$ time,
and subsequently
\{$\sigma(1, i) \mid 1 \le i \le n + 1$\}
can be computed from 
\{$\sigma(0, i) \mid 0 \le i \le n + 1$\}
in $O(n \log n)$ time.
After these two preliminary steps,
we can find $\lambda^*$ by a binary search,
which amounts to first computing $\sigma_\lambda(n + 1)$ and $d_\lambda(n + 1)$,
and then checking whether $d_\lambda(n + 1) \le d^*$,
for $O(\log n)$ different $\lambda$ between $0$ and $n$.

For each such $\lambda$,
we can compute the $n + 1$ pairs $(\sigma_\lambda(i), d_\lambda(i))$, $1 \le i \le n + 1$,
sequentially, by dynamic programming.
For the base case, let $(\sigma_\lambda(1), -d_\lambda(1)) = (\sigma(1, 1) - \lambda, -1)$.
Then, for $j = 2,\ldots,n+1$,
use the recurrence
$$
(\sigma_\lambda(j), -d_\lambda(j)) = \max\big\{\;
(\sigma(1, j) - \lambda, -1),\;
\max\{\,
( \sigma_\lambda(i) + 1 - \inv(i, j)\lambda,\, -d_\lambda(i) - \inv(i, j) )
\mid
1 \le i < j
\,\}
\;\big\}.
$$
Again, with balanced search trees, this can be done in $O(n \log n)$ time for each $\lambda$.
Thus the overall running time is $O(n \log^2 n)$.

\subsubsection{Proof of Lemma~\ref{lem:concave}}

Recall $dp(h, i)$ and $(d,p)\to i$ defined earlier.
For $0 \le i \le n$, let
$$
DP_i = \{\, dp(h, i) \mid 0 \le h \le i \,\}.
$$

We first prove some easy lemmas.
The following lemma is about the lexicographic order of pairs:

\begin{lemma}\label{lem:pair}
	For all integers $a$, $b$, $c$, and $d$,
	$(a, b) < (c, d) \iff (a - 1, b) < (c - 1, d)$.
\end{lemma}

\begin{proof}
	In lexicographic order, $(a, b) < (c, d)$ is equivalent to
	$a < c$ or $a = c$ and $b < d$,
	which is equivalent to
	$a - 1 < c - 1$ or $a - 1 = c - 1$ and $b < d$,
	and hence $(a - 1, b) < (c - 1, d)$.
\end{proof}

The next two lemmas are on basic properties of $(d, p)\to i$:

\begin{lemma}\label{lem:dptoi}
	For $1 \le i \le n$, the $p$-part of $(d, p) \to i$ is $\pi_i$,
	and $(d, p) \to i$ is the smallest pair
	greater than $(d, p)$ with $p$-part equal to $\pi_i$.
\end{lemma}

\begin{proof}
	Recall that $(d, p) \to i$ yields either $(d, \pi_i)$, if $p < \pi_i$,
	or $(d + 1, \pi_i)$, otherwise.
	The greedy choice of $d$ or $d + 1$ ensures that it is the smallest
	among all pairs $(a, \pi_i)$ that are greater than $(d, p)$.
\end{proof}

\begin{lemma}\label{lem:dptoi1}
	For $1 \le i \le n$,
	if $(d, p) < (d', p')\to i$,
	then $(d - 1, p) < (d' - 1, p')\to i$.
\end{lemma}

\begin{proof}
	If $p' < \pi_i$,
	then $(d', p') \to i = (d', \pi_i)$,
	and $(d' - 1, p') \to i = (d' - 1, \pi_i)$.
	If $p' \ge \pi_i$,
	then $(d', p') \to i = (d' + 1, \pi_i)$,
	and $(d' - 1, p') \to i = (d', \pi_i)$.
	Write $(d', p')\to i$ as $(a, \pi_i)$.
	Then $(d' - 1, p') \to i = (a - 1, \pi_i)$.
	By Lemma~\ref{lem:pair},
	if $(d, p) < (a, \pi_i)$,
	then $(d - 1, p) < (a - 1, \pi_i)$.
\end{proof}

The next few lemmas are on basic properties of $dp(h,i)$ and $DP_i$:

\begin{lemma}\label{lem:dph}
	For $0 \le i \le n$,
	$dp(h, i)$ is strictly increasing for $0 \le h \le i$.
\end{lemma}

\begin{proof}
	It suffices to show that
	$dp(h, i) < dp(h + 1, i)$
	for $0 \le h < i$.
	For any sequence $\bs s$ of $h + 2$ indices,
	$0 = s_0 < \ldots < s_h < s_{h+1}\le i$,
	with exactly $d(h+1, i)$ drops and with
	$\pi_{s_{h+1}} = p(h+1, i)$,
	let $\bs s'$ be the subsequence of the first $h + 1$ indices in $\bs s$,
	$0 = s_0 < \ldots < s_h \le i$.
	Then either
	$\pi_{s_h} < \pi_{s_{h+1}}$,
	and $\bs s$ has the same number of drops as $\bs s'$,
	or
	$\pi_{s_h} > \pi_{s_{h+1}}$,
	and $\bs s$ has one more drop than $\bs s'$.
\end{proof}

\begin{lemma}\label{lem:dp0i}
	For $0 \le i \le n$,
	$dp(0, i) = (0, \pi_0)$
	is the smallest pair in $DP_i$,
	and is the only pair in $DP_i$ with $p$-part equal to $\pi_0$.
\end{lemma}

\begin{proof}
	By Lemma~\ref{lem:dph},
	$dp(0, i)$ is the smallest pair in $DP_i$.
	In the definition of $dp(h, i)$,
	the sequence $\bs s$ of $h + 1$ distinct indices $s_0,\ldots,s_h$
	always starts with $s_0 = 0$,
	so the last index $s_h$ can be $0$ only if $h = 0$.
	This implies that the $p$ part of $dp(h, i)$ is equal to $\pi_0$ for $h = 0$,
	and is not equal to $\pi_0$ for $0 < h \le i$.
	Thus $dp(0, i) = (0, \pi_0)$ is the only pair in $DP_i$ with $p$-part equal to $\pi_0$.
\end{proof}

\begin{lemma}\label{lem:dpii}
	For $0 \le i \le n$,
	$dp(i, i)$ is the largest pair in $DP_i$,
	and its $p$-part is equal to $\pi_i$.
	For $1 \le i \le n$, $dp(i, i)$ is the smallest pair
	greater than $dp(i - 1, i - 1)$ with $p$-part equal to $\pi_i$.
\end{lemma}

\begin{proof}
	By Lemma~\ref{lem:dph},
	$dp(i, i)$ is the largest pair in $DP_i$, for $0 \le i \le n$.
	Recall that
	$dp(0, 0) = (0, \pi_0)$,
	and
	$dp(i, i) = dp(i - 1, i - 1) \to i$ for $1 \le i \le n$.
	By Lemma~\ref{lem:dptoi},
	the $p$-part of $dp(i, i)$ is $\pi_i$,
	and $dp(i, i)$ is the smallest
	among all pairs $(a, \pi_i)$ that are greater than $dp(i - 1, i - 1)$,
	for $1 \le i \le n$.
\end{proof}

\begin{lemma}\label{lem:dph_}
	For $0 < h \le i \le n$,
	if $dp(h, i) > (d, p) \in DP_i$,
	then $dp(h - 1, i) \ge (d, p)$.
\end{lemma}

\begin{proof}
	Suppose that $dp(h, i) > (d, p) \in DP_i$.
	Since $(d, p) \in DP_i$, we have $(d, p) = dp(h', i)$ for some $h'$, $0 \le h' \le i$.
	Then $dp(h', i) = (d, p) < dp(h, i)$.
	By Lemma~\ref{lem:dph}, it follows that $h' < h$, and hence $h - 1 \ge h'$.
	By Lemma~\ref{lem:dph} again,
	we have $dp(h - 1, i) \ge dp(h', i) = (d, p)$.
\end{proof}

We are now ready to prove Lemma~\ref{lem:concave} that
$\sigma(d, n + 1)$ is concave for $1 \le d \le n + 1$,
that is,
$\sigma(d + 1, n + 1) - \sigma(d, n + 1)$
is decreasing for $1 \le d \le n$.

Fix any $d$, where $1 \le d \le n + 1$.
From the definitions of $dp(h,i)$ for $0 \le h \le i \le n$
and $\sigma(d,i)$ for $0 \le d \le i \le n + 1$,
we can see that
$\sigma(d, n + 1)$
is the maximum $h$, $0 \le h \le n$,
such that
the $d$-part of $dp(h, n) \to n + 1$ is at most $d$.
By Lemma~\ref{lem:dph},
$\sigma(d, n + 1) + 1$
is equal to the number of distinct values for $h$, $0 \le h \le n$,
such that
the $d$-part of $dp(h, n) \to n + 1$ is at most $d$.
Note that the $d$-part of $dp(h, n) \to n + 1$ is at most $d$
if and only if $dp(h, n) < (d, \pi_{n + 1})$.
Then
$\sigma(d, n + 1) + 1$ is equal to
the number of pairs in $DP_n$ that are less than $(d, \pi_{n + 1})$.

Therefore,
for $1 \le d \le n$,
$\sigma(d + 1, n + 1) - \sigma(d, n + 1)$ is equal to
the number of pairs $(a,p) \in DP_n$ such that
$(d, \pi_{n + 1}) \le (a,p) < (d + 1, \pi_{n + 1})$.
Note that no pair $(a,p) \in DP_n$ can have $p = \pi_{n + 1}$.
For $1 \le d \le n + 1$, define
$k_<(d)$ (respectively, $k_>(d)$)
as the number of pairs $(d, p) \in DP_n$ with
$p < \pi_{n + 1}$ (respectively, $p > \pi_{n + 1}$).
Then $\sigma(d + 1, n + 1) - \sigma(d, n + 1) = k_>(d) + k_<(d + 1)$
for $1 \le d \le n$.

To show that
$\sigma(d + 1, n + 1) - \sigma(d, n + 1)$
is decreasing for $1 \le d \le n$,
it suffices to show that both $k_<(d)$ and $k_>(d)$ are decreasing for $1 \le d \le n + 1$. 
We need the following lemma:

\begin{lemma}\label{lem:inclusion}
	For $0 \le i \le n$,
	if $(d + 1, p) \in DP_i$
	and $(d, p) \ge (0, \pi_0)$,
	then $(d, p) \in DP_i$.
\end{lemma}

Now fix any $d$, where $1 \le d \le n$.
Then $(d, p) > (0, \pi_0)$ for all $p$.
By Lemma~\ref{lem:inclusion} with $i = n$,
any pair $(d + 1, p) \in DP_n$,
either $p < \pi_{n + 1}$ or $p > \pi_{n + 1}$,
has a unique corresponding pair $(d, p) \in DP_n$.
Thus $k_<(d + 1) \le k_<(d)$ and  $k_>(d + 1) \le k_>(d)$.

\bigskip
To complete the proof of Lemma~\ref{lem:concave},
it remains to prove Lemma~\ref{lem:inclusion}.
Our proof is by induction on $i$.
For the base case when $i = 0$,
$dp(h, i)$ is defined only for $h = 0$.
Since $dp(0, 0) = (0, \pi_0)$ is included in $DP_0$,
the lemma clearly holds.

We now proceed to the inductive step, and fix $i > 0$.
By Lemma~\ref{lem:dp0i},
$dp(0, i) = (0, \pi_0)$ is the only pair in $DP_i$ with $p$-part equal to $\pi_0$.
Thus the lemma holds for $p = \pi_0$.
There are two cases remaining:
either $p = \pi_i$,
or $p \neq \pi_i$ and $p \neq \pi_0$.

\paragraph{The first case}

We first consider the case that $p = \pi_i$.
By Lemma~\ref{lem:dpii},
the $p$-part of $dp(i, i)$ is $\pi_i$.

Write $dp(i, i)$ as $(d, \pi_i)$.
By Lemma~\ref{lem:dph},
$dp(i, i)$ is the largest pair in $DP_i$.
It suffices to prove that
$(c, \pi_i) \in DP_i$
for all $c < d$ such that $(c, \pi_i) \ge (0, \pi_0)$.

Fix any $c < d$ such that $(c, \pi_i) \ge (0, \pi_0)$.
Then
$(c, \pi_i) > (0, \pi_0)$
since $\pi_i \neq \pi_0$ for $i > 0$.
By Lemma~\ref{lem:dpii},
$dp(i, i) = (d, \pi_i)$ is the smallest pair
greater than $dp(i - 1, i - 1)$ with $p$-part equal to $\pi_i$.
Since $c < d$,
we must have $(c, \pi_i) \le dp(i - 1, i - 1)$.

Let $h \ge 0$ be the smallest integer such that
$(c, \pi_i) \le \hl{dp(h, i - 1)}$.
Then $h \le i - 1$.
Since $(c, \pi_i) > (0, \pi_0) = dp(0, i - 1)$,
we also have $h > 0$.

Write $dp(h, i - 1)$ as $(d', p')$. Then $(d', p') \in DP_{i - 1}$.
From
$(d', p') \ge (c, \pi_i)$,
it follows by Lemma~\ref{lem:pair} that
$(d' - 1, p') \ge (c - 1, \pi_i)$.

We have
$(d', p') \in DP_{i - 1}$,
$dp(h, i - 1) = (d', p') > (d' - 1, p')$,
and $h > 0$.

We next show that
$\hl{dp(h - 1, i - 1)} \ge (d' - 1, p')$.
Consider two cases:
\begin{itemize}\setlength\itemsep{0pt}

		\item $(d' - 1, p') < (0, \pi_0)$.
			Since $h > 0$, it follows by Lemma~\ref{lem:dph} that
			$dp(h - 1, i - 1) \ge dp(0, i - 1) = (0, \pi_0) > (d' - 1, p')$.

		\item $(d' - 1, p') \ge (0, \pi_0)$.
			Since $(d', p') \in DP_{i - 1}$ and $(d' - 1, p') \ge (0, \pi_0)$,
			it follows by induction that $(d' - 1, p') \in DP_{i - 1}$.
			Since $dp(h, i - 1) > (d' - 1, p')$,
			it follows by Lemma~\ref{lem:dph_} that $dp(h - 1, i - 1) \ge (d' - 1, p')$.

\end{itemize}

Thus
$dp(h - 1, i - 1) \ge (d' - 1, p') \ge (c - 1, \pi_i)$.
On the other hand, our choice of $h$ implies that
$dp(h - 1, i - 1) < (c, \pi_i) \le dp(h, i - 1)$.
Therefore, $dp(h - 1, i - 1)\to i = (c, \pi_i)$.

Recall the following recurrence for $0 < h < i \le n$:
\begin{equation}\label{eq:recurrence}
	dp(h,i) = \min\{\, dp(h, i-1),\, dp(h-1, i-1) \to i \,\}.
\end{equation}
Thus
$dp(h,i) = dp(h - 1, i - 1)\to i = (c, \pi_i)$.
Thus $(c, \pi_i) \in DP_i$.

\paragraph{The second case}

We next consider the case that $p \neq \pi_i$ and $p \neq \pi_0$.
Suppose that $(d, p) \in DP_i$ and $(d - 1, p) \ge (0, \pi_0)$.
We will show that $(d - 1, p) \in DP_i$.

Since $(d, p) \in DP_i$,
it follows that $(d, p) = dp(h', i)$ for some $h'$,
where $0 \le h' \le i$.
By Lemma~\ref{lem:dpii},
the $p$-part of $dp(i, i)$ is equal to $\pi_i$.
By Lemma~\ref{lem:dp0i},
the $p$-part of $dp(0, i)$ is equal to $\pi_0$.
Since $p \neq \pi_i$ and $p \neq \pi_0$,
we must have $0 < h' < i$.

By~\eqref{eq:recurrence},
we have $(d, p) = dp(h', i) = \min \{ dp(h', i - 1), dp(h' - 1, i - 1)\to i \}$.
By Lemma~\ref{lem:dptoi},
the $p$-part of $dp(h' - 1, i - 1)\to i$ is $\pi_i$.
Since $p \neq \pi_i$,
it follows that
$(d, p) \neq dp(h' - 1, i - 1)\to i$.
Thus
$(d, p) = dp(h', i - 1) \in DP_{i - 1}$,
and
$(d, p) < dp(h' - 1, i - 1)\to i$.

Write $dp(h' - 1, i - 1)$ as $(d', p')$. Then $(d', p') \in DP_{i - 1}$.
Also, $(d, p) < (d', p')\to i$.
By Lemma~\ref{lem:dptoi1}, it follows that
$(d - 1, p) < (d' - 1, p')\to i$.
By Lemma~\ref{lem:dph}, we have
$dp(h', i - 1) > dp(h' - 1, i - 1)$,
that is,
$(d, p) > (d', p')$.
By Lemma~\ref{lem:pair}, it follows that
$(d - 1, p) > (d' - 1, p')$.

Since $(d, p) \in DP_{i - 1}$ and $(d - 1, p) \ge (0, \pi_0)$,
it follows by induction that $(d - 1, p) \in DP_{i - 1}$.
Thus $(d - 1, p) = \hl{dp(h, i - 1)}$ for some $h$,
where $0 \le h \le i - 1$.
By Lemma~\ref{lem:dp0i},
the $p$-part of $dp(0, i - 1)$ is $\pi_0$.
Since $p \neq \pi_0$,
we must have $h > 0$.
Thus $0 < h < i$.

We have
$(d', p') \in DP_{i - 1}$,
$dp(h, i - 1) = (d - 1, p) > (d' - 1, p')$,
and $h > 0$.

We next show that
$\hl{dp(h - 1, i - 1)} \ge (d' - 1, p')$.
Consider two cases:
\begin{itemize}\setlength\itemsep{0pt}

		\item $(d' - 1, p') < (0, \pi_0)$.
			Since $h > 0$, it follows by Lemma~\ref{lem:dph} that
			$dp(h - 1, i - 1) \ge dp(0, i - 1) = (0, \pi_0) > (d' - 1, p')$.

		\item $(d' - 1, p') \ge (0, \pi_0)$.
			Since $(d', p') \in DP_{i - 1}$ and $(d' - 1, p') \ge (0, \pi_0)$,
			it follows by induction that $(d' - 1, p') \in DP_{i - 1}$.
			Since $dp(h, i - 1) > (d' - 1, p')$,
			it follows by Lemma~\ref{lem:dph_} that $dp(h - 1, i - 1) \ge (d' - 1, p')$.

\end{itemize}

Thus
$dp(h - 1, i - 1) \ge (d' - 1, p')$.
By Lemma~\ref{lem:dptoi}, it follows that
$dp(h - 1, i - 1)\to i \ge (d' - 1, p')\to i$.
Recall that
$dp(h, i - 1) = (d - 1, p) < (d' - 1, p')\to i$.

By~\eqref{eq:recurrence},
we have
$dp(h, i) = dp(h, i - 1) = (d - 1, p)$.
Thus $(d - 1, p) \in DP_i$, as desired.

\bigskip
This completes the proof of Lemma~\ref{lem:inclusion},
Lemma~\ref{lem:concave}, and Theorem~\ref{thm:pc1}.

\section{Algorithms on intervals of different lengths}

In this section we prove
Proposition~\ref{prp:fpt}.
Recall that \textsc{Join} is a special case of \textsc{J-Pack},
and \textsc{Tile} is a special case of the four problems
\textsc{Pack},
\textsc{Cover},
\textsc{J-Pack},
and
\textsc{J-Cover}.
It suffices to present algorithms for these four problems.

Suppose that $\kappa > 1$.
Let $\bs{\ell} = (\ell_1,\ldots,\ell_\kappa)$
be the $\kappa$ distinct lengths of the $n$ intervals in $\I$.
In the following, we use the notation $\bs{u}$ as a shorthand for
a $\kappa$-tuple $(u_1,\ldots,u_\kappa)$.
In particular,
$\bs{0} = (0,\ldots,0)$
and
$\bs{\tau} = (\tau,\ldots,\tau)$.
For $\bs{u} = (u_1,\ldots,u_\kappa)$,
denote by $-\bs{u}$
the $\kappa$-tuple $(-u_1,\ldots,-u_\kappa)$.
For
$\bs{u} = (u_1,\ldots,u_\kappa)$ and $\bs{v} = (v_1,\ldots,v_\kappa)$,
denote by $\bs{u}-\bs{v}$
the $\kappa$-tuple $(u_1 - v_1,\ldots,u_\kappa - v_\kappa)$,
denote by $\bs{u}\cdot\bs{v}$
the dot product $\sum_{h=1}^\kappa u_h v_h$,
and write $\bs{u} \le \bs{v}$ (respectively, $\bs{u} \ge \bs{v}$, $\bs{u} = \bs{v}$)
if $u_h \le v_h$ (respectively, $u_h \ge v_h$, $u_h = v_h$)
for all $1 \le h \le \kappa$.
Denote by $|\bs{u}|$ the sum $\sum_{h=1}^\kappa |u_h|$.
For $\bs{u}\ge\bs{0}$,
we use the phrase ``$\bs{u}$ intervals'' to refer to
$|\bs{u}|$ intervals, including $u_h$ intervals of each length $\ell_h$, $1 \le h \le \kappa$.

Let $B = [0,\ell_B)$.
Let $I_1,\ldots,I_n$,
where $I_i = [x_i, y_i)$,
be the $n$ intervals in $\I$ sorted in increasing $x_i + y_i$.
Add two dummy intervals
$I_0$ and $I_{n+1}$,
where $I_0 = [x_0, y_0)$,
with $y_0 = 0$ and $x_0 + y_0 < x_1 + y_1$,
and $I_{n+1} = [x_{n+1}, y_{n+1})$,
with $x_{n+1} = \ell_B$ and $x_n + y_n < x_{n+1} + y_{n+1}$.

For $0 \le i \le n+1$, denote by $\I_i$ the subfamily of intervals $I_0,I_1,\ldots,I_i$.
For $0 \le i \le n+1$,
denote by $\bs{m}_i$
the $\kappa$-tuple
$(m_{i,1},\ldots,m_{i,\kappa})$
where
$m_{i,h}$
for $1 \le h \le \kappa$
is the multiplicity of $\ell_h$
(that is, the number of intervals of length $\ell_h$)
in $\I_i \setminus\{I_0,I_{n+1}\}$.

\paragraph{Algorithm for \textsc{Pack}}

Let $A(i, \bs{u}, \bs{v})$,
where $0 \le i \le n+1$,
$\bs{0} \le \bs{u} \le \bs{\tau}$,
and
$\bs{0} \le \bs{v} \le \bs{\tau}$,
be the predicate whether
there exists a subfamily $\J \subseteq \I_i$ of
\emph{pairwise-disjoint}
intervals including $I_0$ and $I_i$,
such that
\begin{itemize}\setlength\itemsep{0pt}

		\item
			$\I_i \setminus \J$ includes \emph{exactly} $\bs{u}$ intervals,

		\item
			the positive gaps between consecutive intervals in $\J$ can accommodate
			$\bs{v}$ intervals (that is,
			the $\bs{v}$ intervals can be partitioned into subfamilies, one subfamily for each gap,
			such that each gap \emph{has space for} the corresponding intervals).

\end{itemize}
Then the $n$ intervals in $\I$ can be packed inside $B$ by moving $\tau$ intervals
if and only if
$A(n+1, \bs{v}, \bs{v})$ is true
for some $\bs{v} \ge \bs{0}$ with
$|\bs{v}| \le \tau$.

The table $A$ can be computed by dynamic programming.
For the base case when $i = 0$,
$A(0, \bs{u}, \bs{v})$ is true if and only if $\bs{u} = \bs{v} = \bs{0}$.
For $1 \le i \le n + 1$,
$A(i, \bs{u}, \bs{v})$ is true if and only if
$A(i', \bs{u'}, \bs{v'})$ is true
for some
$0 \le i' < i$,
$\bs{0} \le \bs{u'} \le \bs{u}$,
and
$\bs{0} \le \bs{v'} \le \bs{v}$
such that
\begin{gather*}
	i' \ge i - 1 - \tau,\quad
	\bs{u} - \bs{u'} = \bs{m}_{i-1} - \bs{m}_{i'},\\
	(\bs{v} - \bs{v'})\cdot\bs{\ell} \le x_i - y_{i'}.
\end{gather*}

\paragraph{Algorithm for \textsc{J-Pack}}

Let $A(i, \bs{u}, \bs{v})$,
where $0 \le i \le n+1$,
$\bs{0} \le \bs{u} \le \bs{\tau}$,
and
$\bs{0} \le \bs{v} \le \bs{\tau}$,
be the predicate whether
there exists a subfamily $\J \subseteq \I_i$ of
\emph{pairwise-disjoint}
intervals including $I_0$ and $I_i$,
such that
\begin{itemize}\setlength\itemsep{0pt}

		\item
			$\I_i \setminus \J$ includes \emph{exactly} $\bs{u}$ intervals,

		\item
			the positive gaps between consecutive intervals in $\J$ can accommodate
			$\bs{v}$ intervals (that is,
			the $\bs{v}$ intervals can be partitioned into subfamilies, one subfamily for each gap,
			such that the two boundary gaps,
			the one bounded by $I_0$ on the left,
			and the one bounded by $I_{n+1}$ on the right,
			if any,
			\emph{have space for},
			while the other gaps \emph{fit exactly}, the corresponding intervals).

\end{itemize}
Then the $n$ intervals in $\I$ can be joined into a contiguous interval
contained in $B$
by moving $\tau$ intervals
if and only if
$A(n+1, \bs{v}, \bs{v})$ is true
for some $\bs{v} \ge \bs{0}$ with
$|\bs{v}| \le \tau$.

The table $A$ can be computed by dynamic programming.
For the base case when $i = 0$,
$A(0, \bs{u}, \bs{v})$ is true if and only if $\bs{u} = \bs{v} = \bs{0}$.
For $1 \le i \le n + 1$,
$A(i, \bs{u}, \bs{v})$ is true if and only if
$A(i', \bs{u'}, \bs{v'})$ is true
for some
$0 \le i' < i$,
$\bs{0} \le \bs{u'} \le \bs{u}$,
and
$\bs{0} \le \bs{v'} \le \bs{v}$
such that
\begin{gather*}
	i' \ge i - 1 - \tau,\quad
	\bs{u} - \bs{u'} = \bs{m}_{i-1} - \bs{m}_{i'},\\
	\begin{cases}
		(\bs{v} - \bs{v'})\cdot\bs{\ell} = x_i - y_{i'}
		&\textrm{if } 0 < i < i' < n + 1\\
		(\bs{v} - \bs{v'})\cdot\bs{\ell} \le x_i - y_{i'}
		&\textrm{otherwise}.
	\end{cases}
\end{gather*}

\paragraph{Algorithm for \textsc{J-Cover}}

Let $A(i, \bs{u}, \bs{v})$,
where $0 \le i \le n+1$,
$\bs{0} \le \bs{u} \le \bs{\tau}$,
and
$\bs{0} \le \bs{v} \le \bs{\tau}$,
be the predicate whether
there exists a subfamily $\J \subseteq \I_i$ of
intervals including $I_0$ and $I_i$,
such that the intervals in $\J \setminus \{I_0,I_{n+1}\}$
are \emph{pairwise-disjoint}, and moreover,
\begin{itemize}\setlength\itemsep{0pt}

		\item
			$\I_i \setminus \J$ includes \emph{exactly} $\bs{u}$ intervals,

		\item
			the positive gaps between consecutive intervals in $\J$ can accommodate
			$\bs{v}$ intervals (that is,
			the $\bs{v}$ intervals can be partitioned into subfamilies, one subfamily for each gap,
			such that the two boundary gaps,
			the one bounded by $I_0$ on the left,
			and the one bounded by $I_{n+1}$ on the right,
			if any,
			\emph{can be covered by},
			while the other gaps \emph{fit exactly}, the corresponding intervals).

\end{itemize}
Then the $n$ intervals in $\I$ can be joined into a contiguous interval
containing $B$
by moving $\tau$ intervals
if and only if
$A(n+1, \bs{v}, \bs{v})$ is true
for some $\bs{v} \ge \bs{0}$ with
$|\bs{v}| \le \tau$.

The table $A$ can be computed by dynamic programming.
For the base case when $i = 0$,
$A(0, \bs{u}, \bs{v})$ is true if and only if $\bs{u} = \bs{v} = \bs{0}$.
For $1 \le i \le n + 1$,
$A(i, \bs{u}, \bs{v})$ is true if and only if
$A(i', \bs{u'}, \bs{v'})$ is true
for some
$0 \le i' < i$,
$\bs{0} \le \bs{u'} \le \bs{u}$,
and
$\bs{0} \le \bs{v'} \le \bs{v}$
such that
\begin{gather*}
	i' \ge i - 1 - \tau,\quad
	\bs{u} - \bs{u'} = \bs{m}_{i-1} - \bs{m}_{i'},\\
	\begin{cases}
		(\bs{v} - \bs{v'})\cdot\bs{\ell} = x_i - y_{i'}
		&\textrm{if } 0 < i < i' < n + 1\\
		(\bs{v} - \bs{v'})\cdot\bs{\ell} \ge x_i - y_{i'}
		&\textrm{otherwise}.
	\end{cases}
\end{gather*}

\paragraph{Algorithm for \textsc{Cover}}

Let $A(i, \bs{u}, \bs{v})$,
where $0 \le i \le n+1$,
$\bs{0} \le \bs{u} \le \bs{\tau}$,
and
$\bs{0} \le \bs{v} \le \bs{\tau}$,
be the predicate whether
there exists a subfamily $\J \subseteq \I_i$ of
intervals including $I_0$ and $I_i$,
\emph{with no interval properly contained in another},
such that
\begin{itemize}\setlength\itemsep{0pt}

		\item
			$\I_i \setminus \J$ includes \emph{at least} $\bs{u}$ intervals.

		\item
			the positive gaps between consecutive intervals in $\J$ can accommodate
			$\bs{v}$ intervals (that is,
			the $\bs{v}$ intervals can be partitioned into subfamilies, one subfamily for each gap,
			such that each gap \emph{can be covered by} the corresponding intervals).

\end{itemize}
Then $B$ can be covered by the $n$ intervals in $\I$ by moving $\tau$ intervals
if and only if
$A(n+1, \bs{v}, \bs{v})$ is true
for some $\bs{v} \ge \bs{0}$ with
$|\bs{v}| \le \tau$.

The table $A$ can be computed by dynamic programming.
For the base case when $i = 0$,
$A(0, \bs{u}, \bs{v})$ is true if and only if $\bs{u} = \bs{v} = \bs{0}$.
For $1 \le i \le n + 1$,
$A(i, \bs{u}, \bs{v})$ is true if and only if
$A(i', \bs{u'}, \bs{v'})$ is true
for some
$0 \le i' < i$,
$\bs{0} \le \bs{u'} \le \bs{u}$,
and
$\bs{0} \le \bs{v'} \le \bs{v}$
such that
\begin{gather*}
	\bs{u} - \bs{u'} \le \bs{m}_{i-1} - \bs{m}_{i'}\\
	x_{i'} \le x_i \textrm{ and } y_{i'} \le y_i,\quad
	(\bs{v} - \bs{v'})\cdot\bs{\ell} \ge x_i - y_{i'}.
\end{gather*}

\paragraph{Running time analysis}

The number of entries in the table $A$ is
$(n + 2)(\tau + 1)^{2\kappa}$.
But at the end of the algorithm,
we need to check only entries
$A(n+1, \bs{v}, \bs{v})$
for $\bs{v} \ge \bs{0}$ with $|\bs{v}| \le \tau$.
Thus we can restrict
the computation of $A(i, \bs{u}, \bs{v})$ 
to
$\bs{u} \ge \bs{0}$ with $|\bs{u}| \le \tau$
and
$\bs{v} \ge \bs{0}$ with $|\bs{v}| \le \tau$.

Consider the directed graph $G_{\kappa,\tau}$ with a vertex
for each $\kappa$-tuple $\bs{w} \ge \bs{0}$ with $|\bs{w}| \le \tau$,
and an edge from $\bs{w}$ to $\bs{w'}$
if and only if $\bs{w'} \le \bs{w}$.
Then $G_{\kappa,\tau}$ has ${\tau+\kappa\choose\kappa}$ vertices,
$O({\tau+\kappa\choose\kappa}^2)$ edges,
and can be built in
$O({\tau+\kappa\choose\kappa}^2\kappa)$ time.
Thus with some pre-processing, we can reduce the number of relevant table entries to
$(n + 2){\tau+\kappa\choose\kappa}^2$.
For the recurrence,
each entry is computed by looking up at most
$(\tau + 1){\tau+\kappa\choose\kappa}^2$
other entries for
\textsc{Pack} / \textsc{J-Pack} / \textsc{J-Cover},
and at most
$(n + 1){\tau+\kappa\choose\kappa}^2$
other entries for
\textsc{Cover},
with $O(\kappa)$ time on each look-up.
So the overall running time,
including the $O(n\log n)$ time on sorting,
is
$O(n\log n + n\tau\kappa{\tau+\kappa\choose\kappa}^4)$
for
\textsc{Pack} / \textsc{J-Pack} / \textsc{J-Cover},
and is
$O(n^2\kappa{\tau+\kappa\choose\kappa}^4)$
for
\textsc{Cover}.

Since $\tau \le n$ and $\kappa \le n$,
${\tau+\kappa\choose\kappa}$ is bounded by a polynomial in $n$
when either $\kappa$ or $\tau$ is constant.
Thus the overall running time is bounded by a polynomial in $n$
when either $\kappa$ or $\tau$ is constant.

\paragraph{A more careful implementation}

The running time of the algorithms for
\textsc{Pack},
\textsc{J-Pack},
and
\textsc{J-Cover}
can be improved by a more careful implementation.
Let $A'(i, t, \bs{d})$, where
$0 \le i \le n+1$,
$0 \le t \le \tau$,
and 
$-\bs{\tau} \le \bs{d} \le \bs{\tau}$,
be the predicate whether there exist
$\bs{u} \ge \bs{0}$ with $|\bs{u}| \le \tau$
and
$\bs{v} \ge \bs{0}$ with $|\bs{v}| \le \tau$,
such that $|\bs{u}| = t$, $\bs{u} - \bs{v} = \bs{d}$, and
$A(i, \bs{u}, \bs{v})$ is true.
Then
$A(n + 1, \bs{v}, \bs{v})$ is true for some
$\bs{v} \ge \bs{0}$ with $|\bs{v}| \le \tau$
if and only if
$A'(n + 1, t, \bs{0})$ is true for some $t \le \tau$.
We can compute $A'$ by dynamic programming in a similar way as $A$.

Since $\bs{d} = \bs{u} - \bs{v}$, we have $d_h = u_h - v_h$ for each $h$, $1 \le h \le \kappa$.
If $d_h$ is negative, then $|d_h|$ is the number of intervals of length $\ell_h$
in $\I \setminus \I_i$ that need to be moved to the gaps between consecutive intervals in $\J$.
If $d_h$ is positive, then $|d_h|$ is the number of intervals of length $\ell_h$
in $\I_i \setminus \J$ that remain to be moved to later gaps.
In both cases, a nonzero component $d_h$ of $\bs{d}$
signifies a commitment to move $|d_h|$ intervals in $\I \setminus \J$.
The $|d_h|$ moves for different values of $h$ are independent
because they correspond to intervals of different lengths.
Thus we only need to consider $\bs{d}$
with $|\bs{d}| \le \tau$.

There are exactly ${\tau+\kappa\choose\kappa}$ nonnegative tuples $\bs{d} \ge \bs{0}$.
For each such tuple,
the number of nonzero components is at most $\min\{\kappa,\tau\}$,
and there are at most $2^{\min\{\kappa,\tau\}}$ ways to add positive or negative signs
to them.
Thus the number of $\kappa$-tuples $\bs{d}$
with $|\bs{d}| \le \tau$
is at most
${\tau+\kappa\choose\kappa}2^{\min\{\kappa,\tau\}}$.
With some pre-processing,
we can reduce the number of entries of $A'$ to
$O(n\tau){\tau+\kappa\choose\kappa}2^{\min\{\kappa,\tau\}}$,
and correspondingly reduce the number of table look-ups for each entry to
$O(\tau){\tau+\kappa\choose\kappa}2^{\min\{\kappa,\tau\}}$.
Then the overall running time becomes
$O(n\log n + n\cdot\kappa\tau^2{\tau+\kappa\choose\kappa}^2 4^{\min\{\kappa,\tau\}})$
for
\textsc{Pack},
\textsc{J-Pack},
and
\textsc{J-Cover}.

\bigskip
This completes the proof of Proposition~\ref{prp:fpt}.

\section{Intractability}

In this section we prove
Theorem~\ref{thm:w1hard}.
Recall that
\textsc{Tile} is a special case of the four problems
\textsc{Pack}, \textsc{Cover},
\textsc{J-Pack}, and \textsc{J-Cover},
which are equivalent when $\ell_\I = \ell_B$.
Thus it suffices to prove the hardness of the two problems \textsc{Tile} and \textsc{Join}.

\subsection{Strong-NP-hardness and W[1]-hardness with parameter $\sigma$}

As a warm-up exercise,
we first present a simple proof of the strong-NP-hardness,
and W[1]-hardness with parameter $\sigma$,
of the two problems \textsc{Tile} and \textsc{Join}.

Our proof is by a reduction from the strongly NP-hard problem
\textsc{Bin Packing}~\cite[Problem~SR1]{GJ79}.
Given $\hat n$ items of integer lengths $a_i$, $1 \le i \le \hat n$,
and $\hat\kappa$ bins each of integer length $b$,
the problem \textsc{Bin Packing} asks whether
the $\hat n$ items can be packed inside the $\hat\kappa$ bins,
that is, whether
the $\hat n$ items can be partitioned into $\hat\kappa$ subsets,
such that the total length of items in each subset is at most $b$.
Our reduction is from a restricted version of \textsc{Bin Packing}
where all integers $a_i$ and $b$ are encoded in unary and moreover
$\sum_{i=1}^{\hat n} a_i = \hat\kappa\,b$.
\textsc{Bin Packing} is W[1]-hard
with parameter $\hat\kappa$ even for this restricted version~\cite{JKMS13}.
Without loss of generality, we assume that $a_i \le b$ for $1 \le i \le \hat n$.

Our reduction works as follows.
Let $\ell_B = \hat\kappa(b + 1)$.
Put the interval $B$ at $[0,\ell_B)$,
then partition it into $2\hat\kappa$ intervals of alternating lengths $b$ and $1$,
where the $\hat\kappa$ intervals of length $b$ are called \emph{bin intervals},
and the $\hat\kappa$ intervals of length $1$ are called \emph{separator intervals}.
Let $\I$ be a family of $n = \hat n + \hat\kappa$ intervals of total length
$\ell_\I = \sum_{i=1}^{\hat n} a_i + \hat\kappa = \ell_B$,
including $\hat n$ \emph{item intervals} of lengths $a_i$,
all sharing the same left endpoint as $B$,
and the $\hat\kappa$ separator intervals from $B$.
Let $\tau = \hat n - 1$ and $\sigma = \hat\kappa + 1$.

\begin{figure}[htb]
	\centering\includegraphics{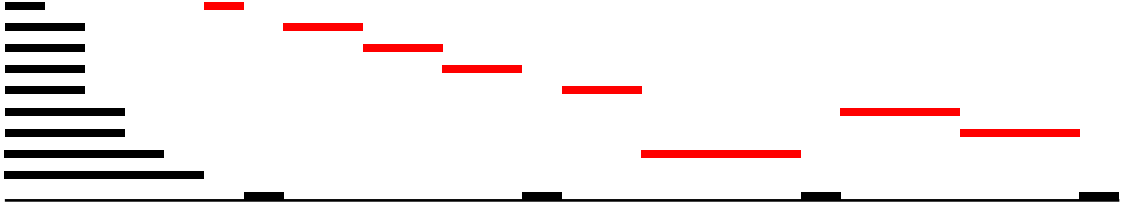}
	\caption{An interval $B$ and a family $\I$ of $n = 13$ intervals
	constructed from a \textsc{Bin Packing} instance consisting of
	$\hat n = 9$ items of lengths $1, 2, 2, 2, 2, 3, 3, 4, 5$
	and $\hat\kappa = 4$ bins of length $6$.
	$B$ can be tiled by moving $\tau = 8$ intervals and keeping $\sigma = 5$ intervals
	unmoved.
	Initial positions of intervals in $\I$ are shown in black;
	final positions of moved intervals are shown in red.
	}\label{fig:sigma}
\end{figure}

This completes the construction.
Refer to Figure~\ref{fig:sigma} for an example.
We claim that
the $\hat n$ items can be partitioned into $\hat\kappa$ subsets
each of total length $b$
if and only if
$B$ can be tiled with $\I$
by moving $\tau$ intervals and keeping $\sigma$ intervals unmoved,
if and only if
$\I$ can be joined into a contiguous interval
by moving $\tau$ intervals and keeping $\sigma$ intervals unmoved.

It is easy to see that
if the $\hat n$ items can be partitioned into $\hat\kappa$ subsets
each of total length $b$,
then $B$ can be tiled with $\I$
by moving $\tau$ intervals and keeping $\sigma$ intervals unmoved.
Also,
any tiling of $B$ with $\I$ necessarily joins $\I$ into a contiguous interval.
Thus we have the two direct, \emph{only if}, implications of the claim.
To complete the proof, it suffices to show that
if $\I$ can be joined into a contiguous interval
by moving $\tau$ intervals and keeping $\sigma$ intervals unmoved,
then the $\hat n$ items can be partitioned into $\hat\kappa$ subsets
each of total length $b$.

Note that the $\hat n$ item intervals pairwise intersect,
but have to become pairwise disjoint when joined into a contiguous interval.
Thus by our choice of $\tau$ and $\sigma$,
all but one of the item intervals must move,
and all separator intervals must not move.
Since there is an unmoved interval in $\I$ at either end of the interval $B$,
and since $\ell_\I = \ell_B$,
the intervals in $\I$ must be joined into
a contiguous interval that coincides with $B$.
It follows that each bin interval of $B$
must be covered by a subset of item intervals of total length $b$.
Correspondingly, we get a partition of the $\hat n$ items in $\hat\kappa$ bins.

The reduction is clearly polynomial,
and is FPT with parameter $\sigma = \hat\kappa + 1$.
Thus \textsc{Tile} and \textsc{Join}
are NP-hard, and W[1]-hard with parameter $\sigma$,
even when the input is encoded in unary.

\subsection{W[1]-hardness with parameter $\kappa$ and with parameter $\tau$}

To prove the W[1]-hardness, with parameter $\kappa$ and with parameter $\tau$,
of the two problems \textsc{Tile} and \textsc{Join},
we use two reductions from the same W[1]-hard problem \textsc{Colored Clique}~\cite{FHRV09}.

Let $G$ be a graph with $\hat n$ vertices and $\hat m$ edges,
where each vertex has one of $\hat\kappa$ colors.
The problem \textsc{Colored Clique}
asks whether there exists in $G$ a \emph{colored clique}
of $\hat\kappa$ pairwise-adjacent vertices
including exactly one vertex of each color.
Denote the $\hat n$ vertices by $0,\dots,\hat n - 1$,
and denote the $\hat\kappa$ colors by $0,\dots,\hat\kappa - 1$.
Without loss of generality, we assume that
every edge in $G$ is incident to two vertices of different colors,
every vertex in $G$ is adjacent to at least one vertex of each of
the other $\hat\kappa - 1$ colors,
and $\hat\kappa \ge 3$.

\subsubsection{W[1]-hardness with parameter $\kappa$}

We first prove the W[1]-hardness with parameter $\kappa$ of \textsc{Tile} and \textsc{Join}.
Let
$\tau = 3\hat\kappa + 3{\hat\kappa\choose 2} + 8(\hat n - 1){\hat\kappa\choose 2}$,
and $\ell = 6\hat\kappa^2$.
We will construct an interval $B$ and a family $\I$ of intervals.
The lengths of the intervals in $\I$ include:
\begin{itemize}\setlength\itemsep{0pt}

		\item
			a \emph{fractional length} $\ell_f = \frac1{\tau + 1}$,

		\item
			a \emph{vertex length} $\ell_v = (2 + 2(\hat n - 1)(\hat\kappa - 1))\cdot\ell$,

		\item
			an \emph{edge length} $\ell_e = (2 + 4(\hat n - 1))\cdot\ell$,

		\item
			a \emph{padding length} $\ell_p = \ell$,

		\item
			two \emph{incidence lengths}
			$\ell_{\imath\jmath}^+ = \ell + \ell_{\imath\jmath}$
			and
			$\ell_{\imath\jmath}^- = \ell - \ell_{\imath\jmath}$,
			where
			$\ell_{\imath\jmath} = \hat\kappa\cdot\imath + \jmath$,
			for each ordered pair of colors $\imath\jmath$ with $\imath \neq \jmath$,

		\item
			two \emph{color lengths}
			$\ell_{\{\imath\}}^+ = \ell + \ell_{\{\imath\}}$
			and
			$\ell_{\{\imath\}}^- = \ell - \ell_{\{\imath\}}$,
			where
			$\ell_{\{\imath\}} = \hat\kappa^2 + \hat\kappa\cdot\imath + \imath$,
			for each color $\imath$,

		\item
			two \emph{color-pair lengths}
			$\ell_{\{\imath\jmath\}}^+ = \ell + \ell_{\{\imath\jmath\}}$
			and
			$\ell_{\{\imath\jmath\}}^- = \ell - \ell_{\{\imath\jmath\}}$,
			where
			$\ell_{\{\imath\jmath\}} = \hat\kappa^2 + \hat\kappa\cdot\imath + \jmath$,
			for each unordered pair of colors $\{\imath,\jmath\}$ with $\imath < \jmath$.

\end{itemize}

In total,
there are $k = 4 + 4{\hat\kappa\choose 2} + 2\hat\kappa + 2{\hat\kappa\choose 2}$ lengths
listed above.
These $k$ lengths are all distinct except that
$\ell_v = \ell_e$ when $\hat\kappa = 3$.
Thus the number $\kappa$ of distinct lengths is either $k$ or $k - 1$,
which is $O(\hat\kappa^2)$.
Since
$\ell_p = \ell = 6\hat\kappa^2$,
$0 < \ell_{\imath\jmath} < \hat\kappa^2$,
$\hat\kappa^2 \le \ell_{\{\imath\}} < 2\hat\kappa^2$,
and
$\hat\kappa^2 < \ell_{\{\imath\jmath\}} < 2\hat\kappa^2$,
the lengths
$\ell_p$,
$\ell_{\imath\jmath}^\pm$,
$\ell_{\{\imath\}}^\pm$,
and
$\ell_{\{\imath\jmath\}}^\pm$
are all greater than $4\hat\kappa^2$ and less than $8\hat\kappa^2$,
and hence differ from $\ell$ and from each other by factors less than two.

\paragraph{The interval $B$ and its partition}

Let
$g_1 = 2\hat\kappa + 2{\hat\kappa\choose 2}$,
$g_2 = 4(\hat n - 1){\hat\kappa\choose 2} + 4(\hat n - 1){\hat\kappa\choose 2}$,
$n_B = (\hat n + \hat m + g_1 + g_2 + 1) + \hat n + \hat m + g_1 + g_2$,
and
$\ell_B = (\hat n + \hat m + g_1 + g_2 + 1)\cdot 1
+ \hat n\cdot\ell_v + \hat m\cdot\ell_e + (g_1 + g_2)\cdot\ell$.
Put the interval $B$ at $[0,\ell_B)$,
then partition it into $n_B$ subintervals,
including
\begin{itemize}\setlength\itemsep{0pt}

		\item
			$\hat n + \hat m + g_1 + g_2 + 1$ \emph{separator intervals} of length $1$,

		\item
			$\hat n$ \emph{vertex intervals} of length $\ell_v$,

		\item
			$\hat m$ \emph{edge intervals} of length $\ell_e$,

		\item
			$g_1$ \emph{type-1 gap intervals} of various lengths with average $\ell$, including
		\begin{itemize}\setlength\itemsep{0pt}

				\item
					one interval of each length
					$\ell_{\{\imath\}}^+$ and $\ell_{\{\imath\}}^-$,
					for each color $\imath$,

				\item
					one interval of each length
					$\ell_{\{\imath\jmath\}}^+$ and $\ell_{\{\imath\jmath\}}^-$,
					for each unordered pair of colors $\{\imath,\jmath\}$ with $\imath < \jmath$,

		\end{itemize}

	\item
		$g_2$ \emph{type-2 gap intervals} of various lengths with average $\ell$, including
	\begin{itemize}\setlength\itemsep{0pt}

			\item
				$\hat n - 1$ intervals of each length
				$\ell_{\imath\jmath}^+$ and $\ell_{\imath\jmath}^-$,
				for each ordered pair of colors $\imath\jmath$ with $\imath \neq \jmath$,
			\item
				$4(\hat n - 1){\hat\kappa\choose 2}$ intervals of length $\ell_p = \ell$,

	\end{itemize}

\end{itemize}
where
the $\hat n + \hat m + g_1 + g_2$
vertex/edge/gap intervals are interspersed between
the $\hat n + \hat m + g_1 + g_2 + 1$
separator intervals.

\paragraph{The family $\I$ of $n$ intervals}

Let $n
= (\hat n + \hat m + g_1 + g_2 + 1)\cdot (\tau + 1)
+ \hat n\cdot(2 + 2(\hat n - 1)(\hat\kappa - 1))
+ \hat m\cdot(2 + 4(\hat n - 1))
+ \hat\kappa + {\hat\kappa\choose 2}$.
Construct the family $\I = \I_0 \cup \I_1 \cup \I_2 \cup \I_3$ of $n$ intervals
in four parts as follows.

First,
for each of the
$\hat n + \hat m + g_1 + g_2 + 1$ separator intervals of length $1$,
partition it further into $\tau + 1$ \emph{fractional intervals}
of length $\ell_f$,
and put them in $\I_0$.

Next, for each vertex $i$ of color $\imath$,
take a distinct vertex interval of length $\ell_v$ in the partition of $B$,
partition it further into $2 + 2(\hat n - 1)(\hat\kappa - 1)$ intervals,
which encode the color $\imath$ and the vertex $i$,
and put them in $\I_1$:
\begin{itemize}\setlength\itemsep{0pt}

		\item
			one interval of each length $\ell_{\{\imath\}}^+$ and $\ell_{\{\imath\}}^-$,

		\item
			$2(\hat n - 1)$ intervals for each color $\jmath \neq \imath$, including
		\begin{itemize}\setlength\itemsep{0pt}

				\item
					$i$ intervals of each length $\ell_{\imath\jmath}^+$ and $\ell_{\imath\jmath}^-$,
				\item
					$\hat n - 1 - i$ pairs of intervals of length $\ell_p = \ell$.

		\end{itemize}

\end{itemize}

Next,
for each edge $\{i,j\}$ of color pair $\{\imath,\jmath\}$ with $\imath < \jmath$,
take a distinct edge interval of length $\ell_e$ in the partition of $B$,
partition it further into $2 + 4(\hat n - 1)$ intervals,
which encode the color pair $\{\imath,\jmath\}$ and the vertices $\{i,j\}$,
and put them in $\I_2$:
\begin{itemize}\setlength\itemsep{0pt}

		\item
			one interval of each length $\ell_{\{\imath\jmath\}}^+$ and $\ell_{\{\imath\jmath\}}^-$,

		\item
			$4(\hat n - 1)$ intervals including
		\begin{itemize}\setlength\itemsep{0pt}

				\item
					$\hat n - 1 - i$ intervals of each length
					$\ell_{\imath\jmath}^+$ and $\ell_{\imath\jmath}^-$,
				\item
					$i$ pairs of intervals of length $\ell_p = \ell$,

				\item
					$\hat n - 1 - j$ intervals of each length
					$\ell_{\jmath\imath}^+$ and $\ell_{\jmath\imath}^-$,
				\item
					$j$ pairs of intervals of length $\ell_p = \ell$.

		\end{itemize}

\end{itemize}

Finally,
construct
$\hat\kappa$ intervals of length $\ell_v$,
and ${\hat\kappa\choose 2}$ intervals of length $\ell_e$,
all sharing the same left endpoint as $B$,
and put them in $\I_3$.

\begin{figure}[htb]
	\centering\includegraphics{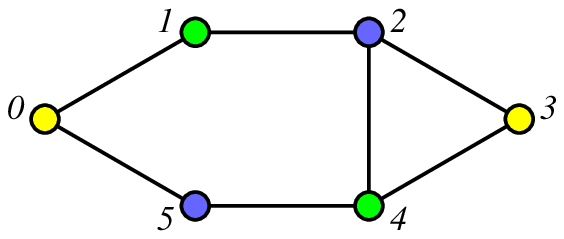}\input{kappa.tex}%
	\begin{align*}
		&	(60,48)^2 \quad (61,47)^2 \quad (55,53)^3 \quad (56,52)^3 \quad (57,51)^4 \quad (59,49)^4\\
		& (60,48)^3 \quad (61,47)^3 \quad (55,53)^2 \quad (56,52)^2 \quad (57,51)^1 \quad (59,49)^1
	\end{align*}
	\caption{Top: A graph of $\hat n = 6$ vertices and $\hat m = 7$ edges with $\hat\kappa = 3$ colors,
	where vertex $i$ has color $i \bmod 3$.
	The three vertices $2,3,4$ form a colored clique.
	Center:
	The $\hat n = 6$ vertex intervals and $\hat m = 7$ edge intervals
	are partitioned into the listed lengths
	(in particular, $\ell_p = 54$ is the padding length),
	where superscripts denote multiplicities,
	and the incidence lengths for the colored clique $2,3,4$ are highlighted in red.
	Bottom:
	The multiplicities of the highlighted lengths from the vertex intervals
	are complementary with those from the edge intervals,
	and add up to $\hat n - 1 = 5$.}\label{fig:kappa}
\end{figure}

\bigskip
This completes the construction of $B$ and $\I$.
Refer to Figure~\ref{fig:kappa} for an example.
Among the intervals in the partition of $B$,
the separator, vertex, and edge intervals
are further partitioned to construct intervals in
$\I_0$, $\I_1$, and $\I_2$, respectively.
Only the gap intervals are not used in the construction of $\I$;
their total length is $(g_1 + g_2)\cdot\ell$.
On the other hand, the total length of the intervals in $\I_3$ is
\begin{align*}
	\hat\kappa\cdot\ell_v + {\hat\kappa\choose 2}\cdot\ell_e
	&= \hat\kappa\cdot(2 + 2(\hat n - 1)(\hat\kappa - 1))\cdot\ell
	+ {\hat\kappa\choose 2}\cdot(2 + 4(\hat n - 1))\cdot\ell
	\\
	&= \left(2\hat\kappa + 2{\hat\kappa\choose 2}
	+ 4(\hat n - 1){\hat\kappa\choose 2}
	+ 4(\hat n - 1){\hat\kappa\choose 2}
	\right)\cdot\ell
	= (g_1 + g_2)\cdot\ell.
\end{align*}
Thus $\ell_\I = \ell_B$, as expected.

For convenience we have used fractional intervals of length $\ell_f = \frac1{\tau + 1}$.
By scaling with factor $\tau + 1$,
the coordinates of all intervals can be converted to integers polynomial in
$\hat\kappa$ and $\hat n$.
Thus the reduction is strongly polynomial.
Since $\kappa = O(\hat\kappa^2)$, the reduction is also FPT\@.
We claim that
$G$ has a colored clique of $\kappa$ vertices
if and only if
$B$ can be tiled with $\I$ by moving $\tau$ intervals,
if and only if
$\I$ can be joined into a contiguous interval in $\tau$ moves.

\paragraph{\textsc{Colored Clique} $\implies$ \textsc{Tile}}

Suppose that $G$ has a colored clique $K$ of $\hat\kappa$ vertices.
We will tile $B$ with $\I$ by moving $\tau$ intervals as follows.
First move the $\hat\kappa + {\hat\kappa\choose 2}$ intervals in $\I_3$
to cover the $\hat\kappa$ vertex intervals 
and the ${\hat\kappa\choose 2}$ edge intervals inside $B$
corresponding to the vertices and edges in $K$.
Next move the corresponding intervals in $\I_1$ and $\I_2$,
including the $2 + 2(\hat n - 1)(\hat\kappa - 1)$ intervals
composing each of these vertex intervals,
and the $2 + 4(\hat n - 1)$ intervals
composing each of these edge intervals,
to cover the $g_1 + g_2$ gap intervals.
The number of these intervals in $\I_1$ and $\I_2$ is
$$
\hat\kappa \cdot (2 + 2(\hat n - 1)(\hat\kappa - 1))
+ {\hat\kappa\choose 2} \cdot (2 + 4(\hat n - 1))
= 2\hat\kappa + 2{\hat\kappa\choose 2}
+ 8(\hat n - 1){\hat\kappa\choose 2}
= g_1 + g_2.
$$
The total number of moves is
$$
\hat\kappa + {\hat\kappa\choose 2}
+ g_1 + g_2
= 3\hat\kappa + 3{\hat\kappa\choose 2} + 8(\hat n - 1){\hat\kappa\choose 2}
= \tau.
$$

Since $K$ is a colored clique,
the intervals in $\I_1$
composing the $\hat\kappa$ vertex intervals
include exactly one pair of intervals of lengths
$\ell_{\{\imath\}}^\pm$
for each color $\imath$,
and the intervals in $\I_2$
composing the ${\hat\kappa\choose 2}$ edge intervals
include exactly one pair of intervals of lengths
$\ell_{\{\imath\jmath\}}^\pm$
for each unordered pair of colors $\{\imath,\jmath\}$ with $\imath < \jmath$,
to cover the
$g_1 = 2\hat\kappa + 2{\hat\kappa\choose 2}$ type-1 gap intervals of these lengths.

Moreover,
by design of the complementary multiplicities of incidence and standard lengths
in partitioning vertex and edge intervals inside $B$ into intervals in $\I_1$ and $\I_2$,
respectively,
the intervals
composing the $\hat\kappa$ vertex intervals
and the ${\hat\kappa\choose 2}$ edge intervals also include,
for each ordered pair of colors $\imath\jmath$ with $\imath \neq \jmath$,
exactly
$\hat n - 1$ pairs of intervals of lengths
$\ell_{\imath\jmath}^\pm$,
and exactly
$\hat n - 1$ pairs of intervals of length $\ell$,
to cover the
$g_2 = 8(\hat n - 1){\hat\kappa\choose 2}$
type-2 gap intervals of these lengths.
Thus $B$ is tiled with $\I$ by moving $\tau$ intervals.

\paragraph{\textsc{Tile} $\implies$ \textsc{Join}}

Any tiling of $B$ with $\I$ necessarily joins $\I$ into a contiguous interval.
Thus if $B$ can be tiled with $\I$ by moving $\tau$ intervals,
then $\I$ can be joined into a contiguous interval in $\tau$ moves.

\paragraph{\textsc{Join} $\implies$ \textsc{Colored Clique}}

Suppose that $\I$ can be joined into a contiguous interval in $\tau$ moves.
We will find a colored clique of $\hat\kappa$ vertices in $G$.

First note that each separator interval inside $B$
is composed of $\tau + 1$ fractional intervals in $\I_0$,
which are so numerous and so short that moving $\tau$ of them is not enough to change
their neighboring spaces by~$1$.
Since all other intervals in $\I$ and all vertex/edge/gap intervals inside $B$
have integer coordinates,
we can assume without loss of generality that no fractional interval in $\I_0$ is moved.
Then, to join the intervals in $\I$ into a contiguous interval,
we must cover all gap intervals inside $B$.

All
$\hat\kappa + {\hat\kappa\choose 2}$
intervals in $\I_3$ contain the $\tau + 1$ fractional intervals in $\I_0$
composing the leftmost separator interval inside $B$,
so they must be moved.
Since their lengths, $\ell_v$ and $\ell_e$,
are greater than the various lengths of the gap intervals inside $B$,
they cannot be moved to fill the gaps directly.
To fill the gaps, we have to use intervals in $\I_1$ and $\I_2$.

Recall that the lengths of all intervals in $\I_1$ and $\I_2$,
and the lengths of all gap intervals inside $B$,
are around $\ell$, and differ from each other by factors less than two.
Thus each gap interval must be covered by one interval of the same length in $\I_1$ or $\I_2$.
In total, we need $g_1 + g_2$ moves to fill the gaps.
Since $\tau = \hat\kappa + {\hat\kappa\choose 2} + g_1 + g_2$,
there is not a single move to waste.
In summary,
we must first move the
$\hat\kappa + {\hat\kappa\choose 2}$
intervals in $\I_3$
to cover some portions of vertex and edge intervals inside $B$,
and then move exactly $g_1 + g_2$ intervals in $\I_1$ and $\I_2$
composing these covered portions
to cover the $g_1 + g_2$ gap intervals inside $B$.

First consider the
$g_1 = 2\hat\kappa + 2{\hat\kappa\choose 2}$
type-1 gap intervals.
Each type-1 gap interval of a color (respectively, color-pair) length
must be covered by one interval of the same length in $\I_1$ (respectively, $\I_2$),
which was constructed from some vertex (respectively, edge) interval inside $B$.
Since there is a type-1 gap interval of every color and color-pair length,
the
$\hat\kappa + {\hat\kappa\choose 2}$
intervals in $\I_3$
must cover exactly $\hat\kappa$ vertex intervals
and ${\hat\kappa\choose 2}$ edge intervals inside $B$,
and the corresponding $\hat\kappa$ vertices and
${\hat\kappa\choose 2}$ edges in $G$ must span all $\hat\kappa$ colors
and all ${\hat\kappa\choose 2}$ color pairs.

Next consider the $g_2 = 8(\hat n - 1){\hat\kappa\choose 2}$ type-2 gap intervals.
Among the intervals in $\I_1$ and $\I_2$ composing
the $\hat\kappa$ vertex intervals
and ${\hat\kappa\choose 2}$ edge intervals inside $B$,
the intervals of color and color-pair lengths are
used to cover the type-1 gap intervals,
so the intervals of incidence and padding lengths are left to cover the type-2 gap intervals.
The $g_2$ type-2 gap intervals include
exactly $\hat n - 1$ type-2 gap intervals of each incidence length,
and $4(\hat n - 1){\hat\kappa\choose 2}$ intervals of the padding length.
In particular,
for each unordered pair of colors $\{\imath,\jmath\}$ with $\imath < \jmath$,
the multiplicities of the four incidence lengths
$\ell_{\imath\jmath}^\pm$ and $\ell_{\jmath\imath}^\pm$
are all $\hat n - 1$.
By design of the complementary multiplicities of incidence lengths in $\I_1$ and $\I_2$,
the target of $\hat n - 1$
for
$\ell_{\imath\jmath}^\pm$ and $\ell_{\jmath\imath}^\pm$
can be reached only if
the multiplicities of $\hat n - 1 - i$ for
$\ell_{\imath\jmath}^\pm$
and
$\hat n - 1 - j$ for
$\ell_{\jmath\imath}^\pm$
in $\I_2$
are paired with
the multiplicities of $i$ for
$\ell_{\imath\jmath}^\pm$
and
$j$ for
$\ell_{\jmath\imath}^\pm$
in $\I_1$.
Such pairings require
the edge $\{i,j\}$ of color pair $\{\imath,\jmath\}$
be consistent with the vertices $i'$ and $j'$
of the corresponding colors $\imath$ and $\jmath$,
that is, $i = i'$ and $j = j'$.
Thus
the $\hat\kappa$ vertices and the ${\hat\kappa\choose 2}$ edges
form a colored clique.

\subsubsection{W[1]-hardness with parameter $\tau$}

We next prove the W[1]-hardness with parameter $\tau$ of \textsc{Tile} and \textsc{Join}.
Let
$\tau = 3\hat\kappa + 3{\hat\kappa\choose 2} + 8{\hat\kappa\choose 2}$,
$\kappa = 4 + 4\hat n(\hat\kappa - 1) + 2\hat\kappa + 2{\hat\kappa\choose 2}$,
and $\ell = \hat n\hat\kappa + \hat\kappa^2$.
We will construct an interval $B$,
and a family $\I$ of intervals of $\kappa$ distinct lengths.
The $\kappa$ lengths include
\begin{itemize}\setlength\itemsep{0pt}

		\item
			a \emph{fractional length} $\ell_f = \frac1{\tau + 1}$,

		\item
			a \emph{vertex length} $\ell_v = 2\cdot 3\ell + 2(\hat\kappa - 1)\cdot 7\ell$,

		\item
			an \emph{edge length} $\ell_e = 2\cdot 3\ell + 4\cdot 5\ell$,

		\item
			a \emph{pairing length} $\ell_p = 12\ell$,

		\item
			four \emph{incidence lengths}
			$\dot\ell_{i\jmath}^+ = 7\ell + \ell_{i\jmath}$,
			$\dot\ell_{i\jmath}^- = 7\ell - \ell_{i\jmath}$,
			$\ddot\ell_{i\jmath}^+ = 5\ell + \ell_{i\jmath}$,
			$\ddot\ell_{i\jmath}^- = 5\ell - \ell_{i\jmath}$,
			where
			$\ell_{i\jmath} = \hat\kappa\cdot i + \jmath + 1$,
			for each vertex $i$ of color $\imath$
			and for each color $\jmath \neq \imath$,

		\item
			two \emph{color lengths}
			$\ell_{\{\imath\}}^+ = 3\ell + \ell_{\{\imath\}}$
			and
			$\ell_{\{\imath\}}^- = 3\ell - \ell_{\{\imath\}}$,
			where
			$\ell_{\{\imath\}} = \hat\kappa\cdot\imath + \imath + 1$,
			for each color $\imath$,

		\item
			two \emph{color-pair lengths}
			$\ell_{\{\imath\jmath\}}^+ = 3\ell + \ell_{\{\imath\jmath\}}$
			and
			$\ell_{\{\imath\jmath\}}^- = 3\ell - \ell_{\{\imath\jmath\}}$,
			where
			$\ell_{\{\imath\jmath\}} = \hat\kappa\cdot\imath + \jmath + 1$,
			for each unordered pair of colors $\{\imath,\jmath\}$ with $\imath < \jmath$.

\end{itemize}

These $\kappa$ lengths are all distinct.
From
$\ell = \hat n\hat\kappa + \hat\kappa^2$,
it follows that
$\hat n\hat\kappa < \ell$
and
$\hat\kappa^2 < \ell$.
Since
$0 < \ell_{i\jmath} \le \hat n\hat\kappa$,
we have
$4\ell <
\ddot\ell_{i\jmath}^\pm
< 6\ell <
\dot\ell_{i\jmath}^\pm
< 8\ell$.
Since
$0 < \ell_{\{\imath\}} \le \hat\kappa^2$
and
$0 < \ell_{\{\imath\jmath\}} < \hat\kappa^2$,
the lengths
$\ell_{\{\imath\}}^\pm$
and
$\ell_{\{\imath\jmath\}}^\pm$
are all greater than $2\ell$ and less than $4\ell$,
and hence differ from each other by factors less than two.

\paragraph{The interval $B$ and its partition}

Let
$g_1 = 2\hat\kappa + 2{\hat\kappa\choose 2}$,
$g_2 = 4{\hat\kappa\choose 2}$,
$n_B = (\hat n + \hat m + g_1 + g_2 + 1) + \hat n + \hat m + g_1 + g_2$,
and
$\ell_B = (\hat n + \hat m + g_1 + g_2 + 1)\cdot 1
+ \hat n\cdot\ell_v + \hat m\cdot\ell_e
+ g_1\cdot 3\ell + g_2\cdot 12\ell$.
Put the interval $B$ at $[0,\ell_B)$,
then partition it into $n_B$ subintervals,
including
\begin{itemize}\setlength\itemsep{0pt}

		\item
			$\hat n + \hat m + g_1 + g_2 + 1$ \emph{separator intervals} of length $1$,

		\item
			$\hat n$ \emph{vertex intervals} of length $\ell_v$,

		\item
			$\hat m$ \emph{edge intervals} of length $\ell_e$,

		\item
			$g_1$ \emph{type-1 gap intervals} of various lengths with average $3\ell$, including
		\begin{itemize}\setlength\itemsep{0pt}

				\item
					one interval of each length
					$\ell_{\{\imath\}}^+$ and $\ell_{\{\imath\}}^-$,
					for each color $\imath$,

				\item
					one interval of each length
					$\ell_{\{\imath\jmath\}}^+$ and $\ell_{\{\imath\jmath\}}^-$,
					for each unordered pair of colors $\{\imath,\jmath\}$ with $\imath < \jmath$,

		\end{itemize}

	\item
		$g_2$ \emph{type-2 gap intervals} of length $\ell_p = 12\ell$,

\end{itemize}
where
the $\hat n + \hat m + g_1 + g_2$
vertex/edge/gap intervals are interspersed between
the $\hat n + \hat m + g_1 + g_2 + 1$
separator intervals.

\paragraph{The family $\I$ of $n$ intervals}

Let $n
= (\hat n + \hat m + g_1 + g_2 + 1)\cdot (\tau + 1)
+ \hat n\cdot (2 + 2(\hat\kappa - 1))
+ \hat m\cdot (2 + 4)
+ \hat\kappa + {\hat\kappa\choose 2}$.
Construct the family $\I = \I_0 \cup \I_1 \cup \I_2 \cup \I_3$ of $n$ intervals
in four parts as follows.

First,
for each of the
$\hat n + \hat m + g_1 + g_2 + 1$ separator intervals of length $1$,
partition it further into $\tau + 1$ \emph{fractional separator intervals}
of length $\ell_f$,
and put them in $\I_0$.

Next, for each vertex $i$ of color $\imath$,
take a distinct vertex interval of length $\ell_v$ in the partition of $B$,
partition it further into $2 + 2(\hat\kappa - 1)$ intervals,
which encode the color $\imath$ and the vertex $i$,
and put them in $\I_1$:
\begin{itemize}\setlength\itemsep{0pt}

		\item
			one interval of each length
			$\ell_{\{\imath\}}^+$ and $\ell_{\{\imath\}}^-$,

		\item
			two intervals for each color $\jmath \neq \imath$, including
		\begin{itemize}\setlength\itemsep{0pt}

				\item
					one interval of each length
					$\dot\ell_{i\jmath}^+$ and $\dot\ell_{i\jmath}^-$.

		\end{itemize}

\end{itemize}

Next,
for each edge $\{i,j\}$ of color pair $\{\imath,\jmath\}$,
take a distinct edge interval of length $\ell_e$ in the partition of $B$,
partition it further into $2 + 4$ intervals,
and put them in $\I_2$:
\begin{itemize}\setlength\itemsep{0pt}

		\item
			one interval of each length
			$\ell_{\{\imath\jmath\}}^+$ and $\ell_{\{\imath\jmath\}}^-$,

		\item
			four intervals including
		\begin{itemize}\setlength\itemsep{0pt}

				\item
					one interval of each length
					$\ddot\ell_{i\jmath}^+$ and $\ddot\ell_{i\jmath}^-$,

				\item
					one interval of each length
					$\ddot\ell_{j\imath}^+$ and $\ddot\ell_{j\imath}^-$.

		\end{itemize}

\end{itemize}

Finally,
construct
$\hat\kappa$ intervals of length $\ell_v$,
and ${\hat\kappa\choose 2}$ intervals of length $\ell_e$,
all sharing the same left endpoint as $B$,
and put them in $\I_3$.

\begin{figure}[htb]
	\centering\includegraphics{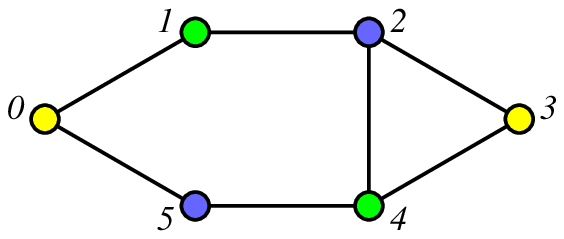}\input{tau.tex}%
	\begin{align*}
		324 &
		= 196 + 128
		= 197 + 127
		= 200 + 124
		= 201 + 123
		= 202 + 122
		= 204 + 120\\
		&
		= 182 + 142
		= 181 + 143
		= 178 + 146
		= 177 + 147
		= 176 + 148
		= 174 + 150
	\end{align*}
	\caption{Top: A graph of $\hat n = 6$ vertices and $\hat m = 7$ edges with $\hat\kappa = 3$ colors,
	where vertex $i$ has color $i \bmod 3$.
	The three vertices $2,3,4$ form a colored clique.
	Center:
	The $\hat n = 6$ vertex intervals and $\hat m = 7$ edge intervals
	are partitioned into the listed lengths,
	where the incidence lengths for the colored clique $2,3,4$ are highlighted in red.
	Bottom:
	The highlighted lengths are complementary,
	and add up to the pairing length $\ell_p = 324$.}\label{fig:tau}
\end{figure}

\bigskip
This completes the construction of $B$ and $\I$.
Refer to Figure~\ref{fig:tau} for an example.
Among the intervals in the partition of $B$,
the separator, vertex, and edge intervals
are further partitioned to construct intervals in
$\I_0$, $\I_1$, and $\I_2$, respectively.
Only the gap intervals are not used in the construction of $\I$;
their total length is $g_1\cdot 3\ell + g_2\cdot 12\ell$.
On the other hand, the total length of the intervals in $\I_3$ is
\begin{align*}
	\hat\kappa\cdot\ell_v + {\hat\kappa\choose 2}\cdot\ell_e
	&= \hat\kappa\cdot(2\cdot 3\ell + 2(\hat\kappa - 1)\cdot 5\ell)
	+ {\hat\kappa\choose 2}\cdot(2\cdot 3\ell + 4\cdot 7\ell)
	\\
	&= \left(2\hat\kappa + 2{\hat\kappa\choose 2}\right)\cdot 3\ell
	+ 4{\hat\kappa\choose 2}\cdot 12\ell
	= g_1\cdot 3\ell + g_2\cdot 12\ell.
\end{align*}
Thus $\ell_\I = \ell_B$, as expected.

For convenience we have used fractional intervals of length $\ell_f = \frac1{\tau + 1}$.
By scaling with factor $\tau + 1$,
the coordinates of all intervals can be converted to integers polynomial in
$\hat\kappa$ and $\hat n$.
Thus the reduction is strongly polynomial.
Since $\tau = O(\hat\kappa^2)$, the reduction is also FPT\@.
We claim that
$G$ has a colored clique of $\kappa$ vertices
if and only if
$B$ can be tiled with $\I$ by moving $\tau$ intervals,
if and only if
$\I$ can be joined into a contiguous interval in $\tau$ moves.

\paragraph{\textsc{Colored Clique} $\implies$ \textsc{Tile}}

Suppose that $G$ has a colored clique $K$ of $\hat\kappa$ vertices.
We will tile $B$ with $\I$ by moving $\tau$ intervals as follows.
First move the $\hat\kappa + {\hat\kappa\choose 2}$ intervals in $\I_3$
to cover the $\hat\kappa$ vertex intervals 
and the ${\hat\kappa\choose 2}$ edge intervals inside $B$
corresponding to the vertices and edges in $K$.
Next move the corresponding intervals in $\I_1$ and $\I_2$,
including the $2 + 2(\hat\kappa - 1)$ intervals
composing each of these vertex intervals,
and the $2 + 4$ intervals
composing each of these edge intervals,
to cover the $g_1 + g_2$ gap intervals:
one interval
of length $\ell_{\{\imath\}}^\pm$ or $\ell_{\{\imath\jmath\}}^\pm$
for each type-1 gap interval;
two intervals
of complementary lengths $\dot\ell_{i\jmath}^+ + \ddot\ell_{i\jmath}^- = \ell_p$
or $\dot\ell_{i\jmath}^- + \ddot\ell_{i\jmath}^+ = \ell_p$
for each type-2 gap interval.
The number of these intervals in $\I_1$ and $\I_2$ is
$$
\hat\kappa
\cdot
(2 + 2(\hat\kappa - 1))
+
{\hat\kappa\choose 2}
\cdot
(2 + 4)
= 2\hat\kappa + 2{\hat\kappa\choose 2} + 2\cdot 4{\hat\kappa\choose 2}
= g_1 + 2 g_2.
$$
The total number of moves is
$$
\hat\kappa + {\hat\kappa\choose 2}
+ g_1 + 2 g_2
= 3\hat\kappa + 3{\hat\kappa\choose 2} + 8{\hat\kappa\choose 2}
= \tau.
$$

Since $K$ is a colored clique,
the intervals in $\I_1$
composing the $\hat\kappa$ vertex intervals
include one pair of intervals of lengths
$\ell_{\{\imath\}}^\pm$
for each color $\imath$,
and the intervals in $\I_2$
composing the ${\hat\kappa\choose 2}$ edge intervals
include exactly one pair of intervals of lengths
$\ell_{\{\imath\jmath\}}^\pm$
for each unordered pair of colors $\{\imath,\jmath\}$ with $\imath < \jmath$,
to cover the
$g_1 = 2\hat\kappa + 2{\hat\kappa\choose 2}$ type-1 gap intervals of these lengths.

Moreover,
by design of the complementary incidence lengths
in partitioning vertex and edge intervals inside $B$ into intervals in $\I_1$ and $\I_2$,
respectively,
the intervals
composing the $\hat\kappa$ vertex intervals
and the ${\hat\kappa\choose 2}$ edge intervals also include,
for each edge $\{i,j\}$ of color pair $\{\imath,\jmath\}$ in $K$,
one interval of each length
$\dot\ell_{i\jmath}^\pm,
\ddot\ell_{i\jmath}^\pm,
\dot\ell_{j\imath}^\pm,
\ddot\ell_{j\imath}^\pm$,
forming four complementary pairs
$\dot\ell_{i\jmath}^+ + \ddot\ell_{i\jmath}^- =
\dot\ell_{i\jmath}^- + \ddot\ell_{i\jmath}^+ =
\dot\ell_{j\imath}^+ + \ddot\ell_{j\imath}^- =
\dot\ell_{j\imath}^- + \ddot\ell_{j\imath}^+ = \ell_p$,
to cover the
$g_2 = 4{\hat\kappa\choose 2}$
type-2 gap intervals of length $\ell_p$.
Thus $B$ is tiled with $\I$ by moving $\tau$ intervals.

\paragraph{\textsc{Tile} $\implies$ \textsc{Join}}

Any tiling of $B$ with $\I$ necessarily joins $\I$ into a contiguous interval.
Thus if $B$ can be tiled with $\I$ by moving $\tau$ intervals,
then $\I$ can be joined into a contiguous interval in $\tau$ moves.

\paragraph{\textsc{Join} $\implies$ \textsc{Colored Clique}}

Suppose that $\I$ can be joined into a contiguous interval in $\tau$ moves.
We will find a colored clique of $\hat\kappa$ vertices in $G$.

By the same argument as in the preceding proof for parameter $\kappa$,
we can assume without loss of generality that no fractional interval in $\I_0$ is moved.
Then, to join the intervals in $\I$ into a contiguous interval,
we must cover all gap intervals inside $B$.

All
$\hat\kappa + {\hat\kappa\choose 2}$
intervals in $\I_3$ contain the $\tau + 1$ fractional intervals in $\I_0$
composing the leftmost separator interval inside $B$,
so they must be moved.
Since their lengths, $\ell_v$ and $\ell_e$,
are greater than the various lengths of the gap intervals inside $B$,
they cannot be moved to fill the gaps directly.
To fill the gaps, we have to use intervals in $\I_1$ and $\I_2$.

Recall that the color lengths $\ell_{\{\imath\}}^\pm$
and the color-pair lengths
$\ell_{\{\imath\jmath\}}^\pm$
differ from each other by factors less than two,
and are smaller than the incidence lengths
$\dot\ell_{i\jmath}^\pm$ and $\ddot\ell_{i\jmath}^\pm$.
Thus each type-1 gap interval must be covered by one interval of the same length in $\I_1$ or $\I_2$.
Also recall that
the pairing length $\ell_p$ of the type-2 gap intervals is greater than
the color lengths, the color-pair lengths, and the incidence lengths
of the intervals in $\I_1$ and $\I_2$.
Thus each type-2 gap interval requires two moves to cover.
In total, we need $g_1 + 2 g_2$ moves to fill the gaps.
Since $\tau = \hat\kappa + {\hat\kappa\choose 2} + g_1 + 2 g_2$,
there is not a single move to waste.
In summary,
we must first move the
$\hat\kappa + {\hat\kappa\choose 2}$
intervals in $\I_3$
to cover some portions of vertex and edge intervals inside $B$,
and then move exactly $g_1 + 2 g_2$ intervals in $\I_1$ and $\I_2$
composing these covered portions
to cover the $g_1 + g_2$ gap intervals inside $B$.

First consider the
$g_1 = 2\hat\kappa + 2{\hat\kappa\choose 2}$
type-1 gap intervals.
Each type-1 gap interval of a color (respectively, color-pair) length
must be covered by one interval of the same length in $\I_1$ (respectively, $\I_2$),
which was constructed from some vertex (respectively, edge) interval inside $B$.
Since there is a type-1 gap interval of every color and color-pair length,
the
$\hat\kappa + {\hat\kappa\choose 2}$
intervals in $\I_3$
must cover exactly $\hat\kappa$ vertex intervals
and ${\hat\kappa\choose 2}$ edge intervals inside $B$,
and the corresponding $\hat\kappa$ vertices and
${\hat\kappa\choose 2}$ edges in $G$ must span all $\hat\kappa$ colors
and all ${\hat\kappa\choose 2}$ color pairs.

Next consider the $g_2 = 4{\hat\kappa\choose 2}$ type-2 gap intervals.
Among the intervals in $\I_1$ and $\I_2$ composing
the $\hat\kappa$ vertex intervals
and ${\hat\kappa\choose 2}$ edge intervals inside $B$,
the intervals of color and color-pair lengths are
used to cover the type-1 gap intervals,
so the intervals of incidence lengths are left to cover the type-2 gap intervals.
They include exactly $\hat\kappa\cdot 2(\hat\kappa - 1) = g_2$
intervals of incidence lengths
$\dot\ell_{i\jmath}^\pm$ in $\I_1$,
and exactly
${\hat\kappa\choose 2}\cdot 4 = g_2$
intervals of incidence lengths
$\ddot\ell_{i\jmath}^\pm$ in $\I_2$.
Recall that
$4\ell <
\ddot\ell_{i\jmath}^\pm
< 6\ell <
\dot\ell_{i\jmath}^\pm
< 8\ell$.
Thus each type-2 gap interval, of length $\ell_p = 12\ell$,
must be covered by two intervals of complementary lengths,
either $\dot\ell_{i\jmath}^+ + \ddot\ell_{i\jmath}^- = \ell_p$
or $\dot\ell_{i\jmath}^- + \ddot\ell_{i\jmath}^+ = \ell_p$,
with matching subscript $i\jmath$.
Such pairings require
the edge $\{i,j\}$ of color pair $\{\imath,\jmath\}$
be consistent with the vertices $i'$ and $j'$
of the corresponding colors $\imath$ and $\jmath$,
that is, $i = i'$ and $j = j'$.
Thus
the $\hat\kappa$ vertices and the ${\hat\kappa\choose 2}$ edges
form a colored clique.

\bigskip
This completes the proof of Theorem~\ref{thm:w1hard}.

\bigskip
\begin{remark}
The four problems 
\textsc{Pack},
\textsc{Cover},
\textsc{J-Pack},
and
\textsc{Tile}
have a variant where $B$ is a simple closed curve,
and $\I$ is a family of intervals on the curve $B$.
Since the family $\I$ of intervals are inside the interval $B$ in our reductions
for the proof of Theorem~\ref{thm:w1hard},
these reductions can be easily adapted to the closed-curve variant
and yield similar hardness results.
\end{remark}

\appendix

\newpage
\section{Torty and Shields}

\begin{center}
	time limit per test: 1 second\\
	memory limit per test: 256 megabytes\\
	input: standard input\\
	output: standard output
\end{center}

\bigskip
Torty the sea turtle is fending off invading jellyfish!

Torty has an army of $n$ turtles standing in a line, each holding a shield of length $\ell$.
Initially, the shield of the $i$-th turtle is at the interval $[x_i, x_i + \ell)$.
Torty wants the turtles to arrange their defense positions so that their shields
concatenate into a contiguous interval of length $n\cdot\ell$
inside the battle interval $[0, b)$.

Turtles are serene creatures.
They would rather bask in the sun than move around.
But once moving, they can go any distance.
What is the minimum number of turtles that have to move to form the target configuration?

Torty entrusts this important problem to you. Solve it quickly!

\subsection*{Input}

The first line contains three integers $n$, $\ell$, and $b$, where $1 \le n \le 2\cdot 10^5$,
$1 \le \ell, b \le 10^9$, and $n\cdot\ell \le b$.
The second line contains $n$ integers $x_1, x_2, \ldots, x_n$,
where $-10^9 \le x_i \le 10^9$.
The shields may overlap at their initial positions.

\subsection*{Output}

Print the minimum number of turtles that have to move.

\subsection*{Sample input}

\begin{verbatim}
6 2 13
-1 3 4 5 12 11
\end{verbatim}

\subsection*{Sample output}

\begin{verbatim}
3
\end{verbatim}

\subsection*{Note}
The best way is to move the turtle with shield at $[-1, 1)$ to $[1, 3)$,
the turtle with shield at $[4, 6)$ to $[7, 9)$,
and the turtle with shield at $[12, 14)$ to $[9, 11)$.
Then all shields are joined into a contiguous interval $[1, 13)$
inside the battle interval $[0, 13)$.

\newpage
\section{Blue Puppy and UFOs}

\begin{center}
	time limit per test: 2 seconds\\
	memory limit per test: 256 megabytes\\
	input: standard input\\
	output: standard output
\end{center}

\bigskip
Blue Puppy is looking for his little brother Torty.
It's time to go home, but Torty is still playing among sea turtles and jellyfish
on the beach, which is an interval $[0, b)$.

Blue Puppy has $n$ UFOs hovering above the beach,
where the $i$-th UFO covers an interval $[x_i, x_i + \ell)$
with surveillance cameras.

What is the minimum number of UFOs that have to move so that the $n$ UFOs together cover
the whole length of the beach?

\subsection*{Input}

The first line contains three integers $n$, $\ell$, and $b$, where $1 \le n \le 10^5$,
$1 \le \ell, b \le 10^9$, and $n\cdot\ell \ge b$.
The second line contains $n$ integers $x_1, x_2, \ldots, x_n$,
where $-10^9 \le x_i \le 10^9$.
The intervals covered by the UFOs may overlap at their initial positions.

\subsection*{Output}

Print the minimum number of UFOs that have to move.

\subsection*{Sample input}

\begin{verbatim}
8 2 10
-1 -2 3 4 5 8 9 10
\end{verbatim}

\subsection*{Sample output}

\begin{verbatim}
2
\end{verbatim}

\subsection*{Note}
One of the best ways is to move the UFO at $[4, 6)$ to $[1, 3)$,
and move the UFO at $[10, 12)$ to $[6, 8)$.
Then the $8$ UFOs together cover the interval $[-2, 11)$,
which contains the beach $[0, 10)$.


\begin{thebibliography}{9}

	\bibitem{CQ21}
		S. Compton and B. Qi.
		\emph{Minimum Removals}.
		USA Computing Olympiad Camp 2021, Day 3, Problem 3,
		May 31, 2021.

	\bibitem{FHRV09}
		M. R. Fellows, D. Hermelin, F. Rosamond, and S. Vialette.
		On the parameterized complexity of multiple-interval graph problems.
		\emph{Theoretical Computer Science},
		410:53--61, 2009.

	\bibitem{GJ79}
		M. R. Garey and D. S. Johnson.
		\emph{Computers and Intractability: A Guide to the Theory of NP-completeness}.
		W. H. Freeman and Company, 1979.

	\bibitem{JKMS13}
		K. Jansen, S. Kratsch, D. Marx, and I. Schlotter.
		Bin packing with fixed number of bins revisited.
		\emph{Journal of Computer and System Sciences},
		79:39--49, 2013.

	\bibitem{MNO11}
		M. Mehrandish, L. Narayanan, and J. Opatrny.
		Minimizing the number of sensors moved on line barriers.
		In
		\emph{Proceedings of the 2011 IEEE Wireless Communications and Networking Conference (WCNC'11)},
		pages 653--658, 2011.

\end{thebibliography}
\end{document}